\documentclass[onecolumn,nofootinbib,thightenlines,notitlepage,tightenlines,longbibliography,superscriptaddress,11pt]{revtex4-1} % the option longbibliography implies that the names of the paper are shown 
\newcommand{\pagenumbaa}{1}
\usepackage{graphicx}
\usepackage{amssymb}
\usepackage{amstext}
\usepackage{enumerate}

\usepackage{mathrsfs}
\usepackage{tikz}
\usetikzlibrary{chains}
\usetikzlibrary{fit}
\usepackage{pgflibraryarrows}		%optional
\usepackage{pgflibrarysnakes}		%optional
\usepackage{xcolor}
\usepackage{epsfig}
\usetikzlibrary{shapes.symbols,patterns} % for source symbols
\usepackage{pgfplots}

\usepackage{amsthm}
\usepackage{amsmath}
\usepackage{amsfonts}
\usepackage{amssymb,amstext}
\usepackage{bbm} % for bbm numbers
\usepackage{dsfont}

\usepackage[colorlinks=true,urlcolor=blue, hyperindex,breaklinks=true] {hyperref}
\usepackage{todonotes}

\usepackage{footnote}

\theoremstyle{plain}
\newtheorem{mythm}{Theorem}[section]
\newtheorem{myprop}[mythm]{Proposition}
\newtheorem{mycor}[mythm]{Corollary}
\newtheorem{mylem}[mythm]{Lemma}
\newtheorem{myclaim}[mythm]{Claim}

\theoremstyle{definition}
\newtheorem{mydef}[mythm]{Definition}
\newtheorem{myex}[mythm]{Example}
\newtheorem{myremark}[mythm]{Remark}
\newtheorem{myass}[mythm]{Assumption}

\newcommand{\ket}[1]{\left| #1 \right \rangle}
\newcommand{\bra}[1]{\left \langle #1 \right|}
\newcommand{\braket}[2]{\left \langle  \left. #1 \right| #2 \right\rangle}
\newcommand{\norm}[1]{\left\lVert#1\right\rVert}

\newcommand{\bracket}[1]{ \left \lbrace \text{#1}\right \rbrace }
\newcommand{\dens}{\mathfrak{D}(R)}

\def\Hb{\ensuremath{H_{\rm b}}}
\newcommand{\W}{\mathsf{W}} %for channel law
\newcommand{\Eu}{\mathsf{E}} %for universal encoder
\def\trnorm{\mathrm{tr}}
\def\opnorm{\mathrm{op}}

\def\poly{\mathrm{poly}}
\def\subexp{\mathrm{subexp}}

\newcommand{\Hp}{\mathrm{H}_{+}^{M}} 
\newcommand{\Hop}{\mathrm{H}^{M}} 

\newcommand{\introsection}[1]{\emph{#1}.---}
\newcommand{\Hh}[1]{H\!\left({#1}\right)} %Entropy
\newcommand{\I}[2]{I\!\left({#1},{#2}\right)} %mutual Information
\newcommand{\Prob}[1]{\,{\mathds P} \!\left[#1\right]} %
\newcommand{\E}[1]{\,{\mathds E}\!\left[#1\right]} %Expectation
\newcommand{\EQ}[1]{\,{\mathds E}^{Q}\!\left[#1\right]} %Expectation

\newcommand{\herm}{\ensuremath{^{\dagger}}}

\newcommand{\Tr}[1]{{\rm tr}\left[{#1}\right]} %Trace
\newcommand{\Trp}[2]{{\rm tr}_{{#1}}\!\left({#2}\right)} %fidelity

\DeclareMathOperator{\spec}{spec}
\newcommand{\st}{\textnormal{s.t.}}
\DeclareMathOperator{\id}{id}

\newcommand{\drv}{\ensuremath{\,\mathrm{d}}}

\newcommand{\inprod}[2]{\ensuremath{\left\langle{#1}\vphantom{\big|},\vphantom{\big|}{#2}\right\rangle}}
\newcommand{\transp}{\ensuremath{^{\scriptscriptstyle{\top}}}}
\newcommand{\WW}{\mathcal{W}}

\newcommand{\ii}{\mathrm{i}} %complex i

\newcommand{\e}{\mathrm{e}} % Euler number
\newcommand{\p}{\mathrm{p}} % polynomial function

\newcommand{\R}{\ensuremath{\mathbb{R}}}
\newcommand{\Rp}{\ensuremath{\R_{\geq 0}}}
\newcommand{\Rps}{\ensuremath{\R_{> 0}}}

\newcommand{\Lp}[1]{\mathrm{L}^{#1}}
\newcommand{\lp}[1]{\ell_{#1}}

\allowdisplaybreaks

\usepackage{makecell}

\begin{document}

%%%%%%%%%%%%%%%%%%%%%%%%%%%%%%%%%%%%%%%%%%%%%%%%%%%%%%%%%%%%%%%%%%%%
% Place your title here
%%%%%%%%%%%%%%%%%%%%%%%%%%%%%%%%%%%%%%%%%%%%%%%%%%%%%%%%%%%%%%%%%%%%
\title{Efficient Approximation of Quantum Channel Capacities}

%%%%%%%%%%%%%%%%%%%%%%%%%%%%%%%%%%%%%%%%%%%%%%%%%%%%%%%%%%%%%%%%%%%%
% If using RevTex
%%%%%%%%%%%%%%%%%%%%%%%%%%%%%%%%%%%%%%%%%%%%%%%%%%%%%%%%%%%%%%%%%%%%
 
  \author{David Sutter}
 \email[]{$\bracket{suttedav,\,renner}$@phys.ethz.ch}
 \affiliation{Institute for Theoretical Physics, ETH Zurich, Switzerland}
 
  \author{Tobias Sutter}
 \email[]{$\bracket{sutter,\,mohajerin}$@control.ee.ethz.ch}
 \affiliation{Automatic Control Laboratory, ETH Zurich, Switzerland}
  
 \author{Peyman Mohajerin Esfahani}
 \email[]{$\bracket{sutter,\,mohajerin}$@control.ee.ethz.ch}
 \affiliation{Automatic Control Laboratory, ETH Zurich, Switzerland}

% \author{John Lygeros}
% \email[]{$\bracket{sutter,\,mohajerin,\,lygeros}$@control.ee.ethz.ch}
% \affiliation{Automatic Control Laboratory, ETH Zurich, Switzerland}

   \author{Renato Renner}
 \email[]{$\bracket{suttedav,\,renner}$@phys.ethz.ch}
 \affiliation{Institute for Theoretical Physics, ETH Zurich, Switzerland}
 %%%%%%%%%%%%%%%%%%%%%%%%%%%%%%%%%%%%%%

% \author{a}
% \email[]{a@.ethz.ch}
% \affiliation{Institute for ..., ETH Zurich, Switzerland}
% 
% \author{b}
% \email[]{b@.ethz.ch}
% \affiliation{Institute for ..., ETH Zurich, Switzerland}
% 
% 
% \author{c}
% \email[]{c@.ethz.ch}
% \affiliation{Institute for ..., ETH Zurich, Switzerland}

%%%%%%%%%%%%%%%%%%%%%%%%%%%%%%%%%%%%%%%%%%%%%%%%%%%%%%%%%%%%%%%%%
% Abstract
%%%%%%%%%%%%%%%%%%%%%%%%%%%%%%%%%%%%%%%%%%%%%%%%%%%%%%%%%%%%%%%%%

\begin{abstract}
We propose an iterative method for approximating the capacity of classical-quantum channels with a discrete input alphabet and a finite dimensional output, possibly under additional constraints on the input distribution. Based on duality of convex programming, we derive explicit upper and lower bounds for the capacity. To provide an $\varepsilon$-close estimate to the capacity, the presented algorithm requires $O\big(\tfrac{(N \vee M) M^3 \log(N)^{1/2}}{\varepsilon}\big)$, where $N$ denotes the input alphabet size and $M$ the output dimension. We then generalize the method for the task of approximating the capacity of classical-quantum channels with a bounded continuous input alphabet and a finite dimensional output. For channels with a finite dimensional quantum mechanical input and output, the idea of a universal encoder allows us to approximate the Holevo capacity using the same method. In particular, we show that the problem of approximating the Holevo capacity can be reduced to a multidimensional integration problem.
For families of quantum channels fulfilling a certain assumption we show that the complexity to derive an $\varepsilon$-close solution to the Holevo capacity is subexponential or even polynomial in the problem size.
We provide several examples to illustrate the performance of the approximation scheme in practice.
\end{abstract}

 \maketitle

%Only for RevTEx
 \setcounter{page}{\pagenumbaa}  
 \thispagestyle{plain}

\section{Introduction}
Consider a scenario where a sender wants to transmit information over a noisy channel to a receiver. Information theory says that there exists fundamental quantities called \emph{channel capacities} characterizing the maximal amount of information that can be transmitted on average, asymptotically reliably per channel use \cite{shannon48}. Depending on the channel and allowed auxiliary resources, there exists a variety of different capacities for different communication tasks. An excellent overview can be found in \cite{wilde_book,holevo_book}. For a lot of these tasks, their corresponding capacity can be recast as an optimization problem. Some of them seem to be intrinsically more difficult than others, however in general none of them is straightforward to compute efficiently.

In this article, we focus on two scenarios. First, we consider the task of sending information over a classical-quantum (cq) channel which maps each element of an input alphabet to a finite dimensional quantum state. We do not allow any additional resources such as entanglement shared between the sender and receiver nor feedback. The capacity for this task has been shown in \cite{holevo98,holevo_book,schumacher97} to be the maximization of a quantity called the \emph{Holevo information} over all possible input distributions. For the case of a finite input alphabet this problem is a finite dimensional convex optimization problem. 
 Based on duality of convex programming and smoothing techniques \cite{nesterov05}, we propose a method to efficiently compute tight upper and lower bounds for the capacity of a finite dimensional cq channel. More precisely, the proposed method has an overall computational complexity of finding an $\varepsilon$-solution given by $O(\tfrac{(N \vee M) M^3 \log(N)^{1/2}}{\varepsilon})$, where $N$ denotes the input alphabet size and $M$ is the output dimension. Our method can treat scenarios where there is an additional constraint on the input distribution of the channel.
As our approach is based on the dual problem, it is possible to extend it to cq channels with a continuous bounded input alphabet and a finite dimensional output.

The second scenario we consider in this article is to send classical information over a quantum channel having a finite dimensional input and output. Again we do not allow additional resources such as entanglement shared between the sender and receiver nor feedback. Compared to the setup of a cq channel, this task is much more delicate as one could make use of entangled input states at the encoding. Indeed it has been shown that the classical capacity of a quantum channel is still poorly understood \cite{hastings09} as only a \emph{regularized} expression is known that describes it \cite{holevo98,holevo_book,schumacher97}, which in general is computationally intractable. The best known generic lower bound for the classical capacity of a quantum channel that has a single letter expression is the \emph{Holevo capacity} which is given by a finite dimensional non-convex optimization problem that has been shown to be $\mathsf{NP}$-complete \cite{shor08}. Using the idea of a universal encoder, we show that this problem is equivalent to the capacity of a cq channel with a continuous bounded input alphabet. Thus, we can apply techniques derived for cq channels to compute close upper and lower bounds for the Holevo capacity that coincide when performing an infinite number of iterations. In each iteration step one has to approximate a multidimensional integral. We derive classes of channels for which the Holevo capacity can be approximated up to an arbitrary precision in subexponential or even polynomial time.

Unlike for classical channels where there exists a specific efficient method---the \emph{Blahut-Arimoto algorithm} \cite{blahut72,arimoto72}---to numerically compute the capacity with a known rate of convergence, something similar for cq channels does not exist up to date. In \cite{shor03}, Shor discusses a combinatorial approach to approximate the Holevo capacity, but he does not prove the convergence of his method. There are numerous different ad hoc approaches to efficiently approach the Holevo capacity, where however no convergence guarantees are given \cite{nagaoka98,nagaoka99,osawa01,hayashi05}.

\vspace{3mm}
\introsection{Notation}
The logarithm with basis 2 is denoted by $\log(\cdot)$ and the natural logarithm by $\ln(\cdot)$.
The space of all Hermitian operators in a finite dimensional Hilbert space $\mathcal{H}$ is denoted by $\Hop$, where $M$ is the dimension of $\mathcal{H}$.
The cone of positive semidefinite Hermitian operators is $\Hp$. For $\sigma \in \Hop$ we denote its set of eigenvalues by $\spec(\sigma)=\{\lambda_1(\sigma),\ldots,\lambda_M(\sigma) \}$.
 %On $\Hop$ we consider two norms, the trace norm $\norm{\cdot}_{1}$ and the operator norm $\norm{\cdot}_{\text{op}}$ (see \cite[p.~6]{holevo_book} for further details).
%The space of all Hermitian operators in a finite dimensional Hilbert space $\mathcal{H}$, equipped with the trace norm, becomes the complex Banach space $\mathfrak{J}(\mathcal{H})$. The same space equipped with the operator norm is the Banach space $\mathfrak{B}(\mathcal{H})$ {\cts mathfrak durch mathcal ersetzen}. Note that these two Banach spaces are in mutual duality \cite[p.~7]{holevo_book}. Let $\mathfrak{J}_{+}(\mathcal{H})$ and $\mathfrak{B}_{+}(\mathcal{H})$ denote its positive cones (which are as well in mutual duality).
We denote the set of density operators on a Hilbert space $\mathcal{H}$ by $\mathcal{D}(\mathcal{H}):= \{ \rho\in \Hp \ : \  \Tr{\rho}=1 \}$. We consider cq channels $\W: \mathcal{X} \to \mathcal{D}(\mathcal{H})$, $x \mapsto \rho_x$ having a finite input alphabet $\mathcal{X}=\{ 1,2,\hdots,N \}$ and a finite output dimension $\dim \mathcal{H}=M$. Each symbol $x\in \mathcal{X}$ at the input is mapped to a density operator $\rho_x$ at the output and therefore the channel can be represented by a set of density operators $\{\rho_x\}_{x \in \mathcal{X}}$. The input probability mass function is denoted by the vector $p\in \R^{N}$ where $p_i = \Prob{X=i}$. A possible input cost constraint can be written as $\E{s(X)} = p\transp s\leq S$, where $s\in \R^{N}$ denotes the cost vector and $S\in \Rp$ is the given total cost. We define the standard $n-$simplex as $\Delta_{n}:=\left\{  x\in\R^{n} : x\geq 0, \sum_{i=1}^{n} x_{i}=1\right\}$. For a probability mass function $p \in \Delta_{N}$ we denote the entropy by $H(p):=-\sum_{i=1}^N p_i \log p_i$. The binary entropy function is defined as $\Hb(x):=-x\log(x) - (1-x)\log(1-x)$ with $x \in [0,1]$. For a probability density $p$ supported at a measurable set $B\subset \R$ we denote the differential entropy by $h(p):=-\int_{B} p(x) \log p(x) \drv x$. The von Neumann entropy is defined by $H(\rho_x):=-\Tr{\rho_x \log \rho_x}$ where $\rho_{x} \in \mathcal{D}(\mathcal{H})$ is a density operator. Let $\Phi:\mathcal{B}(\mathcal{H}_A)\to \mathcal{B}(\mathcal{H}_B)$, where $\mathcal{B}(\mathcal{H})$ denotes the space of bounded linear operators in some Hilbert space $\mathcal{H}$ that are equipped with the trace norm, be a quantum channel that is described by a complete positive trace preserving (cptp) map.
We denote the canonical inner product by $\left \langle x,y \right \rangle := x \transp y$ where $x,y \in \R^n$. For two matrices $A,B \in \mathbb{C}^{m\times n}$, we denote the Frobenius inner product by $\left \langle A,B \right \rangle_F := \Tr{A \herm B}$ and the induced Frobenius norm by $\norm{A}_F:=\sqrt{\left \langle A,A \right \rangle_F}$. The trace norm is defined as $\norm{A}_{\trnorm} :=\mathrm{tr}[\sqrt{A^{\dagger}A}]$. The operator norm is denoted by $\norm{A}_{\opnorm}:=\{\sup_X \norm{AX}_F: \, \norm{X}_F = 1 \}$. For a cptp map $\Phi:\mathcal{B}(\mathcal{H}_A)\to \mathcal{B}(\mathcal{H}_B)$ its diamond norm is defined by $\norm{\Phi}_{\diamond}:=\norm{\Phi \otimes \id_{\mathcal{H}_A}}_{\trnorm}$, where $\norm{\cdot}_{\trnorm}$ denotes the trace norm for resources which is defined as $\norm{\Theta}_{\trnorm}:=\max_{\rho \in \mathcal{D}(\mathcal{H}_A)}  \norm{\Theta(\rho)}_{\trnorm}$.
We denote the maximum and minimum between $a$ and $b$ by $a \vee b$ respectively $a \wedge b$.
The symbol $\preccurlyeq$ denotes the semidefinite order on self-adjoint matrices. The identity matrix of appropriate dimension is denoted by $\mathbf{1}$.
An optimization problem $\min_{x\in S \subset\R^{n}}\{ f_{0}(x) \ : \ f_{j}(x)\leq 0, \  j=1,\hdots,m\}$ is called smooth if all $f_{k}(x)$ for $k=0,\hdots, m$ are differentiable. If there is a non-differentiable component $f_{k}(x)$, it is called non-smooth.

 \vspace{3mm}
\introsection{Structure} The remainder of this article is structured as follows. Section~\ref{sec:cq:channel} shows how to efficiently compute tight upper and lower bounds for the capacity of cq channels having a discrete input alphabet. In Section~\ref{sec:cq:contInput} we then show how to extend the methods introduced in Section~\ref{sec:cq:channel} to approximate the capacity of cq channels with a continuous input alphabet. Using the concept of a universal encoder, this allows us to approximate the Holevo capacity of finite dimensional quantum channels as shown in Section~\ref{sec:Holevo}. We conclude in Section~\ref{sec:conclusion} with a summary and possible subjects of further research. In the interest of readability, some of the technical proofs and details are given in the appendices.

%%%%%%%%%%%%%%%%%%%%%%%%%%%%%%%%%%%%%%%%%%%%%%%%%%%%%%%%%%%%%%%%%%
%%%%%%%%%%%%%%%%%%%%%%%%%%%%%%%%%%%%%%%%%%%%%%%%%%%%%%%%%%%%%%%%%%
\section{Capacity of a Discrete-Input Classical-Quantum Channel} \label{sec:cq:channel}
%%%%%%%%%%%%%%%%%%%%%%%%%%%%%%%%%%%%%%%%%%%%%%%%%%%%%%%%%%%%%%%%%%
In this section we show that concepts introduced in \cite{TobiasSutter14} for a purely classical setup can be generalized to compute the capacity of cq channels with a discrete input alphabet and a bounded output. We consider a discrete input alphabet $\mathcal{X}=\left \lbrace 1,\ldots,N \right \rbrace$ and a finite dimensional Hilbert space $\mathcal{H}$ with $\dim \mathcal{H}=:M$. 
The map $\W: \mathcal{X} \to \mathcal{D}(\mathcal{H})$, $x \mapsto \rho_x$, represents a cq channel. Let $s:\mathcal{X}\to \R_+ $ be some function, $p \in \Delta_N$  and consider the input constraint
\begin{equation}
\inprod{p}{s} \leq S, \label{eq:cqConstr}
\end{equation}
where $S$ is some non-negative constant. 
As shown by Holevo, Schumacher and Westmoreland \cite{holevo98,holevo_book,schumacher97}, the capacity of a cq channel $\W$ satisfying the input constraint \eqref{eq:cqConstr} is given by
\begin{align}\label{eq:cqCapacity}
 C_{\mathsf{cq},S}(\W)=\left\{ \begin{array}{ll}
	\underset{p}{\max} 		&\I{p}{\rho}:= \Hh{\sum_{i =1}^N p_i \rho_i} - \sum_{i=1}^N p_i \Hh{\rho_i} \\
			\st 					& \inprod{p}{s} \leq S \\
			                                & p \in\Delta_N.
	\end{array}\right.
\end{align}
To keep the notation simple we consider a single input constraint as the extension to multiple input constraints is straightforward. 

In the following, we reformulate \eqref{eq:cqCapacity} such that it exhibits a well structured dual formulation and show that strong duality holds. We then show how to smooth the objective function of the dual problem such that it can be solved efficiently using a fast gradient method. Doing so leads to an algorithm that iteratively computes lower and upper bounds to the capacity which converge with a given rate. A key concept in our analysis is that the following problem --- called \emph{entropy maximization} --- with $\lambda\in \Hop$ features an analytical solution
\begin{equation} \label{opt:jaynes}
 	\left\{ \begin{array}{lll}
			&\underset{\rho}{\max} 		&\Hh{\rho} + \Tr{\rho \lambda} \\
			&\st 			&  \rho \in \mathcal{D}(\mathcal{H}).
	\end{array} \right.
\end{equation}
\begin{mylem}[Entropy maximization \cite{jaynes57_2}] \label{lem:jaynes}
Let $\rho^{\star}=2^{-\mu \mathbf{1} +\lambda}$, where $\mu$ is chosen such that $\rho^{\star} \in \mathcal{D}(\mathcal{H})$. Then $\rho^{\star}$ uniquely solves \eqref{opt:jaynes}.
\end{mylem}

We next derive the dual problem of \eqref{eq:cqCapacity} and show how to solve it efficiently. We therefore reformulate \eqref{eq:cqCapacity} by introducing an additional decision variable $\sigma:=\sum_{i=1}^N p_i \rho_i$.
\begin{mylem} \label{lem:inputConstraint}
Let $\mathcal{F}:= \arg\max\limits_{p\in\Delta_{N}} \I{p}{\rho}$ and $S_{\max}:=\min \limits_{p \in \mathcal{F}} \inprod{p}{s}$. If $S \geq S_{\max}$, the optimization problem \eqref{eq:cqCapacity} has the same optimal value as
\begin{align}\label{eq:cqCapacityPrimalwithoutConst}
\mathsf{P}:\left\{ \begin{array}{ll}
	\underset{p,\sigma}{\max} 		& \Hh{\sigma} - \sum_{i=1}^N p_i \Hh{\rho_i} \\
			\st 					& \sigma=\sum_{i=1}^N p_i \rho_i\\
			                                & p \in\Delta_N, \, \sigma \in \mathcal{D}(\mathcal{H}).
	\end{array}\right.
\end{align}
If $S<S_{\max}$, the optimization problem \eqref{eq:cqCapacity} has the same optimal value as
 \begin{align}\label{eq:cqCapacityPrimal}
\mathsf{P}:\left\{ \begin{array}{ll}
	\underset{p,\sigma}{\max} 		& \Hh{\sigma} - \sum_{i=1}^N p_i \Hh{\rho_i} \\
			\st 					& \sigma=\sum_{i=1}^N p_i \rho_i\\
			                                & \inprod{p}{s} = S\\
			                                & p \in\Delta_N,  \sigma \in \mathcal{D}(\mathcal{H}).
	\end{array}\right.
\end{align}
\end{mylem}
\begin{proof}
See Appendix~\ref{app:inputCost}.
\end{proof}
Note that the constraint $\sigma \in \mathcal{D}(\mathcal{H})$ in \eqref{eq:cqCapacityPrimalwithoutConst} and \eqref{eq:cqCapacityPrimal} is redundant since $\rho_i \in \mathcal{D}(\mathcal{H})$ and $p \in \Delta_N$ imply that $\sigma \in \mathcal{D}(\mathcal{H})$.
The Lagrange dual program to \eqref{eq:cqCapacityPrimal} is given by
\begin{align}\label{eq:cqCapacityDual}
 \mathsf{D}: \left\{ \begin{array}{ll}
	\underset{\lambda}{\min} 		&G(\lambda) + F(\lambda) \\
			\st 					& \lambda\in \Hop,
	\end{array}\right.
\end{align}
with $F, G: \Hop\to\R$ of the form
\begin{align} \label{eq:GandF}
	G(\lambda)= \left\{ \begin{array}{ll}
			\underset{p}{\max} 		&\sum_{i =1}^N p_i \left(-\Hh{\rho_i}+\Tr{\rho_i \lambda  } \right) \\
			\st					&\inprod{p}{s} = S \\
							& p\in\Delta_{N}
	\end{array} \right.  \textnormal{and}
	\quad 
	F(\lambda)= \left\{ \begin{array}{ll}
			\underset{\sigma}{\max} 		&H(\sigma)-\Tr{\sigma \lambda  } \\
			\st 				& \sigma \in \mathcal{D}(\mathcal{H})
	\end{array}\right. .
\end{align}
Note that since the coupling constraint $ \sigma=\sum_{i =1}^N p_i \rho_i$ in the primal program \eqref{eq:cqCapacityPrimal} is affine, the set of optimal solutions to the dual program \eqref{eq:cqCapacityDual} is nonempty \cite[Prop.~5.3.1]{ref:Bertsekas-09} and as such the optimum is attained. 
The function $G(\lambda)$ is a (parametric) linear program and $F(\lambda)$ is of the form given in Lemma~\ref{lem:jaynes}, i.e., $F(\lambda)$ has a unique optimizer $\sigma^{\star}=2^{-\mu \mathbf{1} -\lambda  }$, where $\mu$ is chosen such that $\sigma^{\star}\in \mathcal{D}(\mathcal{H})$, which gives 
\begin{equation}
\mu = \log \left( \Tr{2^{-\lambda  }} \right). \label{eq:muCQ}
\end{equation}
We thus obtain
\begin{align}
F(\lambda) &= \Hh{\sigma^{\star}}-\Tr{\sigma^{\star} \lambda  } \nonumber \\
 &= - \Tr{2^{-\mu \mathbf{1} - \lambda  } \log \left( 2^{-\mu \mathbf{1} - \lambda  } \right)} - \Tr{2^{-\mu \mathbf{1} - \lambda  } \lambda  }\nonumber  \\
 &=2^{-\mu} \mu \, \Tr{2^{-\lambda  }} \nonumber \\
 &= \log \left( \Tr{ 2^{- \lambda  }} \right), \label{eq:F}
\end{align}
where the last step uses \eqref{eq:muCQ}. The gradient of $F(\lambda)$ is given by \cite[p.~639 ff.]{ref:Ber-09}
\begin{equation}
\nabla F(\lambda) = - \frac{2^{-\lambda }}{\Tr{2^{-\lambda  }}}. \label{eq:dF}
\end{equation}
The following proposition shows that the gradient \eqref{eq:dF} is Lipschitz continuous, which is essential for the optimization algorithm that we will use to solve \eqref{eq:cqCapacityDual}. 
\begin{myprop}[Lipschitz constant of $\nabla F$] \label{prop:lipschitz:constant:nablaF}
The gradient $\nabla F(\lambda)$ as given in \eqref{eq:dF} is Lipschitz continuous with respect to the Frobenius norm with Lipschitz constant $2$.
\end{myprop}
\begin{proof}
To prove the Lipschitz continuity of $\nabla F(\lambda)$, we focus on the representation of $F(\lambda)$ as an optimization problem, given in \eqref{eq:GandF}. According to \cite[Thm.~1]{nesterov05}, the function $\nabla F(\lambda)$ is Lipschitz continuous with Lipschitz constant $L=\tfrac{1}{\kappa}$, where $\kappa$ is the strong convexity parameter of the convex function $\mathcal{D}(\mathcal{H}) \ni\sigma\mapsto -H(\sigma)\in \R$, where according to \cite[Thm.~16]{ref:Kakade-09} $\kappa = \tfrac{1}{2}$.
\end{proof}
Another requirement to solve \eqref{eq:cqCapacityDual} with a specific rate of convergence using a fast gradient method is that the set of feasible optimizers is compact. In order to assure that and to precisely characterize the size of the set of all feasible optimizers (with respect to the Frobenius norm), we need to impose the following assumption on the cq channel $\W$, that we will maintain for the remainder of this article.
\begin{myass}[Regularity] \label{ass:channel:CQ}
$\gamma:=\min\limits_{x\in \mathcal{X}} \min \spec \left( \rho_{x} \right)>0$
\end{myass}
Even though Assumption~\ref{ass:channel:CQ} may seem restrictive at first glance, it holds for a large class of cq channels. Moreover, according to the Fannes-Audenaert inequality \cite{fannes73,audenaert07} the von Neumann entropy is continuous in its argument. Therefore, cq channels having density operators $\rho_x$ that violate Assumption~\ref{ass:channel:CQ} can be avoided by slight perturbations of these density operators.\footnote{See Example~\ref{ex:two} for a numerical illustration.} Furthermore, it can be seen that the mutual information is strictly concave as a function of the input distribution, for a fixed channel under Assumption~\ref{ass:channel:CQ}. This implies uniqueness of the optimal input distribution.
\begin{mylem} \label{lem:compact:set:CQ}
Under Assumption~\ref{ass:channel:CQ}, the dual program \eqref{eq:cqCapacityDual} is equivalent to 
\begin{align*}
\min\limits_{\lambda} \left\{ G(\lambda) + F(\lambda) \ : \  \lambda\in \Lambda \right\},
\end{align*}
where $\Lambda:= \left\{ \lambda\in \Hop \ : \ \norm{\lambda}_{F}\leq M \log\left(\gamma^{-1} \vee \e \right) \right\}$.
\end{mylem}
\begin{proof}
See Appendix~\ref{app:compact}. 
\end{proof}

\begin{mylem}
Strong duality holds between \eqref{eq:cqCapacityPrimal} and \eqref{eq:cqCapacityDual}.
\end{mylem}
\begin{proof}
The assertion follows by a standard strong duality result of convex optimization, see \citep[Proposition~5.3.1, p.~169]{ref:Bertsekas-09}.
\end{proof}

The goal is to efficiently solve \eqref{eq:cqCapacityDual}, which is not straightforward since $G(\cdot)$ is non-smooth and as therefore in general the subgradient method is optimal to solve such problems \cite{ref:nesterov-book-04}. 
The idea is to use the particular structure of \eqref{eq:cqCapacityDual} that allows us to invoke Nesterov's smoothing technique \cite{nesterov05}. Therefore, we consider
\begin{equation} \label{eq:CQ:G_nu}
G_{\nu} (\lambda) :=  \left \lbrace \begin{array}{ll}
\max \limits_{p} & \inprod{p}{b(\lambda)} - \inprod{p}{a} + \nu \Hh{p} - \nu \log N \\
\st & \inprod{p}{s} =S\\
& p \in \Delta_N,
\end{array} \right.
\end{equation}
with smoothing parameter $\nu \in \R_{>0}$ and $a,b(\lambda) \in \R^N$ defined as $a_i := \Hh{\rho_i}$ and $b_i(\lambda) := \Tr{\rho_i \lambda  }$. 
We denote by $p_{\nu}(\lambda)$ the optimal solution that is unique since the objective function is strictly concave. Clearly for any $p\in \Delta_{N}$, $G_{\nu}(\lambda)\leq G(\lambda)\leq G_{\nu}(\lambda) + \nu D_{2}$ for $D_{2}:=\log(N)$, i.e., $G_{\nu}(\lambda)$ is a uniform approximation of the non-smooth function $G(\lambda)$.
According to Lemma~2.2 in \cite{TobiasSutter14} an analytical optimizer $p_{\nu}(\lambda)$ is given by
\begin{align}
p_{\nu}(\lambda)_{i} = 2^{\mu_{1} + \frac{1}{\nu}(b_{i}(\lambda) - a_{i}) + \mu_{2}s_{i}}, \quad 1\leq i \leq N, \label{eq:optimizerP}
\end{align}
where $\mu_{1},\mu_{2}\in\R$ have to be chosen such that $\inprod{p_{\nu}(\lambda)}{s}=S$ and $p_{\nu}(\lambda)\in\Delta_{N}$. 
\begin{myremark}\label{rmk:stabilization:optimizer:CQ}
In case of no input constraints, the unique optimizer to \eqref{eq:CQ:G_nu} is given by
\begin{equation*}
p_{\nu}(\lambda)_{i} = \frac{2^{ \tfrac{1}{\nu} (b_{i}(\lambda)-a_{i})} }{\sum_{j=1}^{N}2^{ \tfrac{1}{\nu}(b_{j}(\lambda)-a_{j})}},\quad 1\leq i \leq N,
\end{equation*}
whose straightforward evaluation is numerically difficult for small $\nu$. A numerically stable method for this computation is presented in \cite[Rmk.~2.6]{TobiasSutter14}.
\end{myremark}
\begin{myremark}[\cite{TobiasSutter14}]\label{rmk:finite:constraint:optimizer:CQ}
In case of an additional input constraint, we need an efficient method to find the coefficients $\mu_{1}$ and $\mu_{2}$ in \eqref{eq:optimizerP}. In particular if there are multiple input constraints (which will lead to multiple $\mu_{i}$) the efficiency of the method computing them becomes important. Instead of solving a system of non-linear equations, it turns out that the $\mu_{i}$ can be found by solving the following convex optimization problem \cite[p.~257 ff.]{ref:Lasserre-11}
\begin{equation} \label{eq:opt:problem:find:mu:finite}
\sup\limits_{\mu\in\R^{2}}\left\{ \inprod{y}{\mu} - \sum_{i=1}^{N}p_{\nu}(\lambda,\mu) \right\},
\end{equation}
where $y:=(1,S)$. Note that $\eqref{eq:opt:problem:find:mu:finite}$ is an unconstrained maximization of a concave function, whose gradient and Hessian can be easily computed, which would allow us to use second-order methods. 
\end{myremark}

Finally, we can show that the uniform approximation $G_{\nu}(\lambda)$ is smooth and has a Lipschitz continuous gradient with known Lipschitz constant. 
\begin{myprop}[Lipschitz constant of $\nabla G_{\nu}$] \label{prop:CQ:lipschitz}
$G_{\nu}(\lambda)$ is well defined and continuously differentiable at any $\lambda\in \Lambda$. Moreover, it is convex and its gradient $\nabla G_{\nu}(\lambda)=\sum_{i=1}^{N} \rho_{i}\herm p_{\nu}(\lambda)_{i}$ is Lipschitz continuous with respect to the Frobenius norm with constant $\tfrac{1}{\nu}$.
\end{myprop}
\begin{proof}
See Appendix~\ref{app:propLip}.
\end{proof}

We consider the smooth, convex optimization problem
\begin{align}\label{Lagrange:Dual:Program:smooth}
 \mathsf{D}_{\nu}:\left\{ \begin{array}{ll}
	\min\limits_{\lambda} 		& F(\lambda) + G_{\nu}(\lambda) \\
			\st					& \lambda\in \Lambda,
	\end{array}\right.
\end{align}
whose objective function has a Lipschitz continuous gradient with respect to the Frobenius norm with Lipschitz constant $L_{\nu}:=2+\tfrac{1}{\nu}$. According to \cite[Thm.~16]{ref:Kakade-09} the function $\Hop\ni A\mapsto d(A):=\tfrac{1}{2}\norm{A}_{F}^{2}\in \Rp$ is $\tfrac{1}{2}$-strongly convex with respect to the Frobenius norm. 
As such $\mathsf{D}_{\nu}$ can be  be approximated with Nesterov's optimal scheme for smooth optimization \cite{nesterov05}, which is summarized in Algorithm~\hyperlink{algo:1}{1}, where $\pi_{\Lambda}$ denotes the projection operator onto the set $\Lambda$, defined in Lemma~\ref{lem:compact:set:CQ}, that is the Frobenius norm ball with radius $r:=M\log\left(\gamma^{-1} \vee \e \right)$.

\begin{myprop}[Projection on Frobenius norm ball]  \label{prop:proj}
Consider the Frobenius norm ball $\Lambda:=\{ A \in \Hop : \norm{\varsigma(A)}_{2}\leq r\}$ of radius $r\geq 0$, where $\varsigma(A)\in\R^{M}$ denotes the vector of singular values of A. The unique projection of a matrix $B\in\Hop$ onto $\Lambda$ in the Frobenius norm is given by
\begin{equation*}
\pi_{\Lambda}(B) = U \mathsf{diag}\left( \pi_{\Lambda}(\varsigma(B)) \right) V\transp, 
\end{equation*}
where $B=U \mathsf{diag}\left( \varsigma(B) \right) V\transp$ is the singular value decomposition of $B$ and $\pi_{\Lambda}$ is the projection operator of the $\ell_{2}$-norm ball of radius $r$, i.e.,
\begin{equation*} 
		\pi_{\Lambda}(x):=\left\{ \begin{array}{ll} r \tfrac{x}{\norm{x}_{2}}, & \norm{x}_{2}>r \\ x, & \text{otherwise.} \end{array} \right. 
\end{equation*}
\end{myprop}
\begin{proof}
The proof follows the lines in \cite[Prop.~5.3]{richter_phd}.
\end{proof}

 \begin{table}[!htb]
\centering 
\begin{tabular}{c}
  \Xhline{3\arrayrulewidth}  \hspace{1mm} \vspace{-3mm}\\ 
\hspace{34.2mm}{\bf{\hypertarget{algo:1}{Algorithm 1: } }} Optimal scheme for smooth optimization for cq channels \hspace{12.5mm} \\ \vspace{-3mm} \\ \hline \vspace{-0.5mm}
\end{tabular} \\
\vspace{-5mm}

 \begin{flushleft}
  {\hspace{1.7mm}Choose some $\lambda_0 \in \Hop$}
 \end{flushleft}
 \vspace{-8mm}

 \begin{flushleft}
  {\hspace{1.7mm}\bf{For $m\geq 0$ do$^{*}$}}
 \end{flushleft}
 \vspace{-8mm}
 
  \begin{tabular}{l l}
{\bf Step 1: } & Compute $\nabla F(\lambda_{m})+\nabla G_{\nu}(\lambda_{m})$ \\
{\bf Step 2: } & $y_m = \pi_{\Lambda}\left(-\frac{1}{L_{\nu}}\left( \nabla F(\lambda_m)+\nabla G_{\nu}(\lambda_m) \right) + \lambda_k\right)$\\
{\bf Step 3: } &   $z_m=\pi_{ \Lambda}\left(-\frac{1}{2 L_{\nu}} \sum_{i=0}^{m} \frac{i+1}{2} \left(  \nabla F(\lambda_i)+\nabla G_{\nu}(\lambda_i) \right)\right)$\\
{\bf Step 4: } & $\lambda_{m+1}=\frac{2}{m+3}z_{m} + \frac{m+1}{m+3}y_{m}$\\ 
\vspace{-5mm}
  \end{tabular}
 \begin{flushleft}
  {\hspace{1.7mm}[*The stopping criterion is explained in Remark~\ref{remark:stopping}]}
  \vspace{-10mm}
 \end{flushleft}  
\begin{tabular}{c}
\hspace{38mm} \phantom{ {\bf{Algorithm:}} Optimal Scheme for Smooth Optimization}\hspace{36.5mm} \\ \vspace{-1.0mm} \\\Xhline{3\arrayrulewidth}
\end{tabular}
\end{table}

The following theorem provides explicit error bounds for the solution of Algorithm~\hyperlink{algo:1}{1} after $k$ iterations.
Note that $D_{1} = \tfrac{1}{2} \left( M \log\left(\gamma^{-1} \vee \e \right)\right)^{2}$ and $D_{2}=\log(N)$.
\begin{mythm}[\cite{nesterov05}] \label{thm:errorCQdiscrete}
Consider a smoothing parameter
\begin{align*}
\nu = \nu(k) = \frac{2}{k+1}\sqrt{\frac{2 D_{1}}{D_{2}}}.
\end{align*}
Then after $k$ iterations we can generate the approximate solutions to the problems \eqref{eq:cqCapacityDual} and \eqref{eq:cqCapacity}, namely,
\begin{align} \label{eq:optInPut}
\hat{\lambda} = y_{k} \in \Lambda \qquad \text{and}\qquad \hat{p}=\sum_{i=0}^{k}\frac{2(i+1)}{(k+1)(k+2)} p_{\nu}(\lambda_{i})\in \Delta_{N},
\end{align}
which satisfy
\begin{align} \label{eq:EBB}
0\leq F(\hat{\lambda}) + G(\hat{\lambda}) - \I{\hat{p}}{\rho} \leq \frac{4}{k+1} \sqrt{2 D_{1} D_{2}} + \frac{16 D_{1}}{(k+1)^{2}}.
\end{align}
Thus, the complexity of finding an $\varepsilon$-solution to the problems \eqref{eq:cqCapacityDual} and \eqref{eq:cqCapacity} does not exceed
\begin{align*}
4  \sqrt{2 D_{1} D_{2}} \ \frac{1}{\varepsilon} + 4 \sqrt{\frac{ D_{1}}{\varepsilon}}
\end{align*}
iterations.
\end{mythm}
Note that Theorem~\ref{thm:errorCQdiscrete} provides an explicit error bound given in \eqref{eq:EBB}, also called \emph{a priori error}. In addition this theorem predicts an approximation to the optimal input distribution \eqref{eq:optInPut}, i.e., the optimizer of the primal problem. Thus, by comparing the values of the primal and the dual optimization problem, one can also compute an \emph{a posteriori error} which is the difference of the dual and the primal problem, namely $F(\hat{\lambda}) + G(\hat{\lambda}) - \I{\hat{p}}{\rho}$ with $C_{\mathsf{cq},\textnormal{UB}}(\W):=F(\hat{\lambda}) + G(\hat{\lambda})$ and $C_{\mathsf{cq},\textnormal{LB}}(\W):= \I{\hat{p}}{\rho}$.
In practice the a posteriori error is often much smaller than the a priori error (see Section~\ref{sec:simulation:results}).

\begin{myremark}[Stopping criterion of Algorithm~\hyperlink{algo:1}{1}] \label{remark:stopping}
There are two immediate approaches to define a stopping criterion for Algorithm~\hyperlink{algo:1}{1}.
\begin{enumerate}[(i)]
\item \emph{A priori stopping criterion}: Choose an a priori error $\varepsilon>0$. Setting the right hand side of \eqref{eq:EBB} equal to $\varepsilon$ defines a number of iterations $k_{\varepsilon}$ that has to be run in order to ensure an $\varepsilon$-close solution.
\item \emph{A posteriori stopping criterion}: Choose an a posteriori error $\varepsilon>0$. Choose the smoothing parameter $\nu(k_{\varepsilon})$ for $k_{\varepsilon}$ as defined above in the a priori stopping criterion. Fix a (small) number of iterations $\ell$ that are run using Algorithm~\hyperlink{algo:1}{1}. Compute the a posteriori error $\mathrm{e}_{\ell}:= F(\hat{\lambda}) + G(\hat{\lambda}) - \I{\hat{p}}{\rho}$ as given by Theorem~\ref{thm:errorCQdiscrete}. If $\mathrm{e}_{\ell}\leq \varepsilon$ terminate the algorithm otherwise continue with another $\ell$ iterations. Continue until the a posteriori error is below $\varepsilon$.
\end{enumerate}
\end{myremark}

\begin{myremark}[No input cost constraint \& numerical stability] \label{remark:comp:stablility}
In the absence of an input cost constraint (i.e., $s(\cdot)=0$), we can derive a closed form expression for $G_{\nu}(\lambda)$ and its gradient. 
Using \eqref{eq:optimizerP} we obtain
\begin{align}
G_{\nu}(\lambda) &= \nu \log \left( \sum_{i=1}^N 2^{\tfrac{1}{\nu}\left( b(\lambda) - a \right)_i} \right) - \nu \log N \nonumber \\
\frac{\partial G_{\nu}(\lambda)}{\partial \lambda_{m,\ell}} &= \left(\nabla G_{\nu}(\lambda) \right)_{m,\ell}= \frac{1}{S(\lambda)} \sum_{i=1}^N 2^{\tfrac{1}{\nu} \left(b(\lambda) - a  \right)_i} (\rho_i)_{\ell,m}, \label{eq:CQ:grad:G}
\end{align}
where $S(\lambda)= \sum_{i=1}^N 2^{\tfrac{1}{\nu} \left(b(\lambda)  - a  \right)_i} $ and we have used $\tfrac{\partial \Tr{\rho \lambda }}{\partial \lambda_{m,\ell}}=\rho_{\ell,m}$ \cite[Prop.~10.7.2]{ref:Ber-09}. Recall that as introduced above we consider $a,b(\lambda) \in \R^N$, such that $a_i = \Hh{\rho_i}$ and $b_i(\lambda) = \Tr{\rho_i \lambda }$.  In order to achieve an $\varepsilon$-precise solution the smoothing factor $\nu$ has to be chosen in the order of $\varepsilon$, according to Theorem~\ref{thm:errorCQdiscrete}. A straightforward computation of $\nabla G_{\nu}(\lambda)$ via \eqref{eq:CQ:grad:G} for a small enough $\nu$ is numerically difficult. In the light of  \cite[p.~148]{nesterov05}, we present a numerically stable technique for computing $\nabla G_{\nu}(\lambda)$. By considering the functions $\R^{M}\ni \lambda \mapsto f(\lambda)=b(\lambda) - a$ and $\R^{N}\ni x \mapsto R_{\nu}(x)=\nu \log \left( \sum_{i=1}^{N}2^{\tfrac{x_{i}}{\nu}} \right)\in\R$ it is clear that $\nabla_{\lambda} R_{\nu}(f(\lambda))=\nabla G_{\nu}(\lambda)$. The basic idea is to define $\tilde{f}(\lambda):=\max_{1\leq i \leq N} f_{i}(\lambda)$ and then consider a function $g:\R^{M}\to \R^{N}$ given by $g_{i}(\lambda)=f_{i}(\lambda)-\tilde{f}(\lambda)$, such that all components of $g(\lambda)$ are non-positive. One can show that
\begin{equation*} 
\nabla_{\lambda} R_{\nu}(f(\lambda))=\nabla_{\lambda} R_{\nu}(g(\lambda))+ \nabla \tilde{f}(\lambda),
\end{equation*}
where the term on the right-hand side can be computed with a small numerical error. 
\end{myremark}

\begin{myremark}[Complexity]
Recall that a singular value decomposition of a matrix $A \in \mathbb{C}^{M \times M}$ can be done with complexity $O(M^3)$ \cite[Lect.~31]{trefethen97}.
A closer look at Algorithm~\hyperlink{algo:1}{1} reveals that the complexity of a single iteration is $O(M^2 (N\vee M))$. Thus by Theorem~\ref{thm:errorCQdiscrete}, the complexity to compute an $\varepsilon$-close solution using Algorithm~\hyperlink{algo:1}{1} is $O\big(\tfrac{(N \vee M) M^3 \log(N)^{1/2}}{\varepsilon}\big)$.
\end{myremark}

%%%%%%%%%%%%%%%%%%%%%%%%%%%%%%%%%
\subsection{Simulation results} \label{sec:simulation:results}
This section presents two examples to illustrate the performance of the approximation method introduced above. We consider two channels which both exhibit an analytical closed form solution for the capacity. The first example is a channel that satisfies Assumption~\ref{ass:channel:CQ}, whereas the second one does not.
To save computation time we have chosen two channels with a binary input alphabet.
All the simulations in this section are performed on a 2.3 GHz Intel Core i7 processor with 8 GB RAM with Matlab.
\begin{myex} \label{ex:one}
Consider a cq channel $\W$ with a binary input alphabet, i.e., $\mathcal{X}=\{0,1\}$, such that $0\mapsto \rho_0=\tfrac{1}{2}\left(\begin{smallmatrix}
1&0\\0&1
\end{smallmatrix}\right)$ and $1\mapsto \rho_1=\tfrac{1}{4} \left( \begin{smallmatrix}
2&1\\1&2
\end{smallmatrix}\right)$.
A simple calculation leads to an analytical expression of the capacity $C_{\mathsf{cq}}(\W)=\Hb\!\left(\tfrac{16}{43}\right)-\tfrac{21}{43}-\tfrac{22}{43}\,\Hb\!\left(\tfrac{1}{4} \right)\approx 0.048821003204$.
Note that $\spec(\rho_0)=\{\tfrac{1}{2},\tfrac{1}{2}\}$ and $\spec(\rho_1)=\{\tfrac{1}{4},\tfrac{3}{4}\}$, which gives $\gamma:=\min_{x\in \mathcal{X}} \min \spec(\rho_x)=\tfrac{1}{4}$. As predicted by Theorem~\ref{thm:errorCQdiscrete}, Algorithm~\hyperlink{algo:1}{1} has the following a priori error bound
\begin{equation*}
0\leq C_{\mathsf{cq},\textnormal{UB}}(\W)-C_{\mathsf{cq},\textnormal{LB}}(\W) \leq \frac{4 \sqrt{2 D_1 D_2}}{k+1} + \frac{16D_1}{(k+1)^2},
\end{equation*}
where $k$ denotes the number of iterations and $D_1=\tfrac{1}{2}(M \log\left(\gamma^{-1} \vee \e \right))^2 = 8$ and $D_2=\log N = 1$. Table~\ref{tab:ex1} shows the performance of Algorithm~\hyperlink{algo:1}{1} for this example.

 \begin{table}[!htb]
\centering 
\caption{Example~\ref{ex:one} with  $D_1=8$ and $D_2=1$ }
\label{tab:ex1}

  \begin{tabular}{c @{\hskip 3mm} | @{\hskip 3mm} c @{\hskip 2mm} c @{\hskip 2mm} c @{\hskip 2mm} c}
 A priori error & $10^{-1}$ & $10^{-2}$ & $10^{-3}$ & $10^{-4}$ \\
 $C_{\mathsf{cq},\textnormal{UB}}(\W)$ & 0.049\,841\,307\,3& 0.048\,972\,899\,3 & 0.048\,837\,263\,6 & 0.048\,822\,641\,1 \\
 $C_{\mathsf{cq},\textnormal{LB}}(\W)$ & 0.048\,820\,977\,3 & 0.048\,820\,982\,7 & 0.048\,821\,003\,3 & 0.048\,821\,003\,6 \\
 A posteriori error & 1.00$\cdot 10^{-3}$ & 1.52$\cdot 10^{-4}$ & 1.63$\cdot 10^{-5}$& 1.64$\cdot 10^{-6}$  \\
 Time [s] & 0.05 & 0.8 & 4.6 & 47  \\
 Iterations & 167 & 1607 & 16\,007 & 160\,007 
  \end{tabular}
\end{table}

\end{myex}

\begin{myex} \label{ex:two}
Consider a cq channel $\W$ with a binary input alphabet, i.e., $\mathcal{X}=\{0,1\}$, such that $0 \mapsto \rho_0=\ket{0}\!\!\bra{0} =\left(\begin{smallmatrix}
1&0\\0&0
\end{smallmatrix}\right)$ and $1 \mapsto \rho_1 = \ket{+}\!\!\bra{+}=\tfrac{1}{2}\left( \begin{smallmatrix}
1 & 1\\1 & 1
\end{smallmatrix} \right)$. The capacity of this channel can be computed to be $C_{\mathsf{cq}}(\W)=\Hb\left( \tfrac{1}{2}(1+\tfrac{1}{\sqrt{2}})\right) \approx 0.600876$. Note that $\spec(\rho_0)=\spec(\rho_1)=\{0,1\}$ which violates Assumption~\ref{ass:channel:CQ}. As mentioned above a possible solution is to perturb the cq channel by some small parameter $\varepsilon \in(0,\tfrac{1}{2})$ such that Assumption~\ref{ass:channel:CQ} is valid. We consider the perturbed cq channel $\tilde{\W}$ that maps $0\mapsto \tilde \rho_0 = \left( \begin{smallmatrix}
1-\varepsilon & 0 \\ 0 & \varepsilon
\end{smallmatrix} \right)$ and $1\mapsto \tilde \rho_1 = \left( \begin{smallmatrix}
\tfrac{1}{2}+\varepsilon & \tfrac{1}{2}-\varepsilon \\ \tfrac{1}{2}-\varepsilon & \tfrac{1}{2}-\varepsilon
\end{smallmatrix} \right)$.
By continuity of the von Neumann entropy \cite{fannes73,audenaert07}, when choosing $\varepsilon$ being small we only change the value of the capacity by a small amount. More precisely, let us consider $\varepsilon = 10^{-10}$. A simple calculation gives 
\begin{equation*}
|C_{\mathsf{cq}}(\W)-C_{\mathsf{cq}}(\tilde \W)| \leq 2.53474\cdot10^{-9}.
\end{equation*}
Using the triangle inequality and Theorem~\ref{thm:errorCQdiscrete}, we can bound the a priori error of Algorithm~\hyperlink{algo:1}{1} as
\begin{align*}
|C_{\mathsf{cq},\textnormal{UB}}(\tilde \W)-C_{\mathsf{cq}}(\W)| &\leq |C_{\mathsf{cq},\textnormal{UB}}(\W)-C_{\mathsf{cq}}(\tilde \W)| + |C_{\mathsf{cq}}(\tilde \W)-C_{\mathsf{cq}}(\W)|\\
&\leq \frac{4 \sqrt{2 D_1 D_2}}{k+1} + \frac{16D_1}{(k+1)^2} +  2.53474\cdot10^{-9}, 
\end{align*}
where $k$ denotes the number of iterations and $D_1=\tfrac{1}{2}(M \log\left(\gamma^{-1} \vee \e \right))^2 \approx 2207.04$ and $D_2=\log N =1$. The a posteriori error is given by $C_{\mathsf{cq},\textnormal{UB}}(\tilde \W)-C_{\mathsf{cq},\textnormal{LB}}(\tilde \W) + 2.53474\cdot10^{-9}$.
 \begin{table}[!htb]
\centering 
\caption{Example~\ref{ex:two} with  $D_1\approx 2207.04$ and $D_2=1$ using a perturbation parameter $\varepsilon=10^{-10}$. }
\label{tab:ex2}

  \begin{tabular}{c @{\hskip 3mm} | @{\hskip 3mm} c @{\hskip 2mm} c @{\hskip 2mm} c @{\hskip 2mm} c }
 A priori error  &   1 & $10^{-1}$ & $10^{-2}$  \\
  $C_{\mathsf{cq},\textnormal{UB}}(\tilde \W)$ & 0.600\,876\,033\,385\,197& 0.600\,876\,033\,316\,571 & 0.600\,876\,033\,316\,571  \\
 $C_{\mathsf{cq},\textnormal{LB}}(\tilde \W)$ & 0.600\,876\,033\,160\,937 & 0.600\,876\,033\,315\,310 & 0.600\,876\,033\,316\,571 \\
  A posteriori error  & 2.54$\cdot10^{-9}$ & 2.53$\cdot10^{-9}$ & 2.53$\cdot10^{-9}$  \\
 Time [s] & 0.1 & 0.8 & 7.9  \\
 Iterations  &  181 &1392 &13\,353 
  \end{tabular}
\end{table}
\end{myex}

%%%%%%%%%%%%%%%%%%%%%%%%%%%%%%%%%%
%%%%%%%%%%%%%%%%%%%%%%%%%%%%%%%%%%
\section{Capacity of a Continuous-Input Classical-Quantum Channel} \label{sec:cq:contInput}
In this section we generalize the approach introduced in Section~\ref{sec:cq:channel} to cq channels having a continuous bounded input alphabet and a finite dimensional output. There are two major challenges compared to the discrete input alphabet setup treated in Section~\ref{sec:cq:channel}. The first difficulty is that the differential entropy is in general not bounded. This makes the smoothing step more difficult and in particular complicates the task of proving an a priori error bound. A second difficulty in the continuous input alphabet setting is the evaluation of the gradient of the Lagrange dual function which involves an integration that can only be computed approximately. Thus the robustness of the iterative protocol needs to be analyzed.\footnote{This point will become especially important in Section~\ref{sec:Holevo}.}

Within this section, we consider cq channels of the form $\W:\mathcal{P}(R )\to \mathcal{D}(\mathcal{H})$, $x\mapsto \rho_x$, where $R$ is a compact subset of the non-negative real line, $\mathcal{P}(R)$ denotes the space of all probability distributions on $R$ and $M:=\dim \mathcal{H} < \infty$. In addition we consider an input constraint of the form\footnote{The extension to multiple average input cost constraints is straightforward.}
\begin{equation}
\inprod{p}{s} = \int_{R} s(x) \,p(\!\drv x) \leq S, \label{eq:inputConstra}
\end{equation}
for $s\in \Lp{\infty}(R)$ and $p\in \mathcal{P}(R)$.
To properly state a formula describing the capacity of the channel $\W$ with an input constraint \eqref{eq:inputConstra}, we need to assume certain regularity conditions on the function $s$. Let $\{\ket{e_i} \}$ be an orthonormal basis in the Hilbert space $\mathcal{H}$ and $\{f_i \}$ a sequence of real numbers bounded from below. The expression
\begin{equation}
K \ket{\psi} = \sum_i f_i  \ket{e_i} \braket{e_i}{\psi}, \label{eq:defK}
\end{equation}
defines a self adjoint operator $K$ on the dense domain 
\begin{equation}
\mathsf{D}(K)=\left \lbrace \psi \in \mathcal{H} : \, \sum_i \left| f_i \right|^2  \left| \braket{e_i}{\psi} \right|^2 < \infty \right \rbrace, \label{eq:denseDomain}
\end{equation}
where $f_i$ are the eigenvalues and $\ket{e_i}$ the corresponding eigenvectors.

\begin{mydef}[{\cite[Def.~11.3]{holevo_book}}]
An operator defined on the domain \eqref{eq:denseDomain} by the formula \eqref{eq:defK} is called \emph{an operator of type $\mathcal{K}$}.
\end{mydef}

\begin{myass}[Assumptions on the input constraint function] \label{ass:continuous}
In the reminder of this section we impose the following assumption on the input constraint function $s:R \to \R$.
\begin{enumerate}[(i)]
\item There exists a self-adjoint operator $K$ of type $\mathcal{K}$ satisfying $\Tr{\exp(-\theta K)}<\infty$ for all $\theta >0$ such that $s(x) \geq \Tr{\rho_x K}$, $x \in R$. \label{item:i}
\item $s$ is lower semicontinuous and for all $k\in \Rp$ the set $\{x: s(x) \leq k \}\subset R$ is compact. \label{item:ii}
\end{enumerate}
\end{myass}
Assumption~\ref{ass:continuous}\eqref{item:i} implies that $\sup_{p \in \mathcal{P}(R)} \Hh{\int_R \rho_x \,p(\!\drv x)}<\infty$ and Assumption~\ref{ass:continuous}\eqref{item:ii} ensures that the set $\{p \in \mathcal{P}(R):\inprod{p}{s}\leq S \}$ is weakly compact \cite[Lem.~11.14]{holevo_book}.
Under Assumption~\ref{ass:continuous}, the capacity of channel $\W$ is given by \cite[Thm.~11.15]{holevo_book}
\begin{align} \label{eq:irrel}
 C_{\mathsf{cq},S}(\W)=\left\{ \begin{array}{ll}
	\underset{p}{\max} 		&\I{p}{\rho}:= \Hh{\int_{R}  \rho_x \,p(\!\drv x)} - \inprod{p}{\Hh{\rho}} \\
			\st 					&\inprod{p}{s} \leq S \\
			                                & p \in\mathcal{P}(R).
	\end{array}\right.
\end{align}
\begin{myprop} \label{lem:density:dense}
The optimization problem \eqref{eq:irrel} is equivalent to
\begin{equation}
C_{\mathsf{cq},S}(\W) = \sup\limits_{p\in\dens} \left\{ \I{p}{\rho} \ : \ \inprod{p}{s} \leq S \right\},\label{eq:cqContCapacity}
\end{equation}
where $\dens$ is the space of probability densities with support $R$, i.e., $\dens:=\{ f\in\Lp{1}(R) \, : \, f\geq 0, \, \int_{R}f(x)\drv x = 1 \}$.
\end{myprop}
\begin{proof} 
The proof follows by the proof of \cite[Prop.~3.4]{TobiasSutter14} and the lower semicontinuity of the von Neumann entropy \cite[Thm.~11.6]{holevo_book}.
\end{proof}

We consider the pair of vector spaces
$(\Lp{1}(R),\Lp{\infty}(R))$ along with the bilinear form
\begin{align*}
\inprod{f}{g}:=\int_{R}f(x)g(x)\drv x.
\end{align*}
In the light of \cite[Thm.~243G]{ref:fremlin-03} this is a dual pair of vector spaces; we refer to \cite[Sec.~3]{anderson87} for the details of the definition of dual pairs of vector spaces. 
Considering the Frobenius inner product as a bilinear form on the dual pair $(\Hop,\Hop)$, we define the linear operator $\WW: \Hop\to \Lp{\infty}(R)$ and its adjoint operator $\WW^{\star}:\Lp{1}(R)\to\Hop$ by
%
%Let $\MM(R)$ be the space of finite signed measures on the Borel $\sigma$-algebra $\Borelsigalg{R }$ and let $\mathcal{C}(R)$ denote the space of continuous functions on $R$. We define the dual pair of vector spaces
%$(\MM(R),\mathcal{C}(R))$ along with the bilinear form
%\begin{align*}
%\inprod{\mu}{u}:=\int_{R}u(x)\drv \mu(x).
%\end{align*}
%Consider the standard inner product as a bilinear form on the dual pair $(\Hop,\Hop)$, we define the linear operator $\WW: \Hop \to \mathcal{C}(R)$ and its adjoint operator $\WW^{\star}:\MM(R)\to\Hop$, given by
\begin{align*}
\WW \lambda (x) := \Tr{\rho_{x}\lambda}, \qquad \quad \WW^{\star}p := \int_{R}\rho_{x}\,p(\!\drv x).
\end{align*}
We next derive the dual problem of \eqref{eq:cqContCapacity} and show how to solve that efficiently. To this end, we introduce an additional decision variable $\sigma:=\WW^{\star}p$ and reformulate problem \eqref{eq:cqContCapacity}.
\begin{mylem}
Let $\mathcal{F}:= \arg \max \limits_{p\in\dens} \I{p}{\rho}$ and $S_{\max}:= \min \limits_{p\in\dens}\inprod{p}{s}$. If $S\geq S_{\max}$ the optimization problem \eqref{eq:cqContCapacity} has the same optimal value as 
\begin{align*}
\mathsf{P}:\left\{ \begin{array}{ll}
	\underset{p,\sigma}{\sup} 		& \Hh{\sigma} -  \inprod{p}{H(\rho)} \\
			\st					& \sigma=\WW^{\star}p\\
			                                & p \in\dens, \, \sigma \in \mathcal{D}(\mathcal{H}).
	\end{array}\right.
\end{align*}
If $S<S_{\max}$ the optimization problem \eqref{eq:cqContCapacity} has the same optimal value
\begin{align}\label{eq:cqContCapacityPrimal}
\mathsf{P}:\left\{ \begin{array}{ll}
	\underset{p,\sigma}{\sup} 		& \Hh{\sigma} -  \inprod{p}{H(\rho)} \\
			\st					& \sigma=\WW^{\star}p\\
			                                		& \inprod{p}{s} = S\\
			                                		& p \in\dens,  \sigma \in \mathcal{D}(\mathcal{H}).
	\end{array}\right.
\end{align}
\end{mylem}
\begin{proof}
Follows by a similar argument as given in Appendix~\ref{app:inputCost} for the finite dimensional input setup.
\end{proof}

The Lagrange dual program to \eqref{eq:cqContCapacityPrimal} is given by
\begin{align}\label{eq:cqContCapacityDual}
 \mathsf{D}: \left\{ \begin{array}{ll}
	\underset{\lambda}{\inf} 		&G(\lambda) + F(\lambda) \\
			\st					& \lambda\in \Hop,
	\end{array}\right.
\end{align}
where $F, G:\Hop\to\R$ are given by
\begin{align*} 
	G(\lambda)= \left\{ \begin{array}{ll}
			\underset{p}{\sup} 		&\inprod{p}{\WW\lambda} - \inprod{p}{H(\rho)} \\
				\st				&\inprod{p}{s} = S \\
		   								& p\in \dens
	\end{array} \right.  \textnormal{and}
	\quad 
	F(\lambda)= \left\{ \begin{array}{ll}
			\underset{\sigma}{\max} 		&H(\sigma)-\Tr{\sigma \lambda  } \\
			\st 				& \sigma \in \mathcal{D}(\mathcal{H})
	\end{array}\right. .
\end{align*}
Note that $G(\lambda)$ is a (parametric) infinite dimensional linear program and $F(\lambda)$ is exactly of the same form as in Section~\ref{sec:cq:channel}. According to \eqref{eq:F} and \eqref{eq:dF} we thus have
\begin{equation} \label{eq:grad:F}
F(\lambda) = \log \left( \Tr{ 2^{- \lambda  }} \right) \quad \textnormal{and} \quad \nabla F(\lambda) = - \frac{2^{-\lambda }}{\Tr{2^{-\lambda  }}}.
\end{equation}
Note that by Proposition~\ref{prop:lipschitz:constant:nablaF}, $ \nabla F(\lambda)$ is Lipschitz continuous with respect to the Frobenius norm with Lipschitz constant $2$. 
\begin{mylem}
Strong duality holds between \eqref{eq:cqContCapacityPrimal} and \eqref{eq:cqContCapacityDual}.
\end{mylem}
\begin{proof}
The lemma follows from the standard strong duality results of convex optimization, see \cite[Thm.~6]{mitter08}.
\end{proof}

In the remainder of this article we impose the following assumption on the cq channel.
\begin{myass}[Assumption on the cq channel] \label{ass:channel:CQcont}
$\gamma:=\min\limits_{x\in R} \min\spec \rho_{x} >0$
\end{myass}
\begin{mylem} \label{lem:compact:set:CQ:cts}
Under Assumption~\ref{ass:channel:CQcont}, the dual program \eqref{eq:cqContCapacityDual} is equivalent to 
\begin{align*}
\min\limits_{\lambda} \left\{ G(\lambda) + F(\lambda) \ : \  \lambda\in \Lambda \right\},
\end{align*}
where $\Lambda:= \left\{ \lambda\in \Hop \ : \ \norm{\lambda}_{F}\leq M \log\left(\gamma^{-1} \vee \e \right) \right\}$. 
\end{mylem}
\begin{proof}
The proof is a direct extension of the one for Lemma~\ref{lem:compact:set:CQ}.
\end{proof}
As a preliminary result, consider the following entropy maximization problem that exhibits an analytical solution
\begin{equation} \label{opt:cover}
 	\left\{ \begin{array}{lll}
			&\underset{p}{\max} 		&h(p) + \inprod{p}{c} \\
			&\st			&\inprod{p}{s} = S\\
			& 					& p\in \dens.
	\end{array} \right.
\end{equation}
\begin{mylem}[{Entropy maximization \cite[Thm.~12.1.1]{cover}}]  \label{lem:coverCont}
Let $p^{\star}(x) =2^{\mu_1 + c(x) + \mu_2 s(x)}$, $x\in R$ where $\mu_1$ and $\mu_2$ are chosen such that $p^{\star}$ satisfies the constraints in \eqref{opt:cover}. Then $p^{\star}$ uniquely solves \eqref{opt:cover}.
\end{mylem}

The goal is to efficiently compute \eqref{eq:cqContCapacityDual} which is not straightforward since $G(\cdot)$ is non-smooth. Similar as in Section~\ref{sec:cq:channel} the idea is to use Nesterov's smoothing technique \cite{nesterov05}. Therefore we consider 
\begin{equation} \label{eq:cqContGnu}
	G_{\nu}(\lambda)= \left\{ \begin{array}{ll}
			\underset{p}{\max} 		&\inprod{p}{\WW\lambda-\Hh{\rho}} + \nu h(p) -\nu \log(\upsilon) \\
						\st			&\inprod{p}{s} = S \\
						& p\in \dens,
	\end{array} \right.
\end{equation}
where $\upsilon:=\int_{R}\drv x$.
 Problem~\eqref{eq:cqContGnu} is of the form given in Lemma~\ref{lem:coverCont} and therefore has a unique optimizer 
\begin{equation} \label{eq:optimizer:CQ:cts}
p_{\nu}^{\lambda}(x)=2^{\mu_1+\frac{1}{\nu}({\rm{tr}}[\rho_x \lambda ] - H(\rho_x)) + \mu_2 s(x)}, \, \, x \in R,
\end{equation}
where $\mu_1, \mu_2$ are chosen such that $p_{\nu}^{\lambda} \in \dens$ and $\inprod{p_{\nu}^{\lambda}}{s}=S$. Recall that $h( p)\leq  \log(\upsilon) $ for all $p\in\dens$ and that there exists a function $\iota:\R_{>0}\to\Rp$ such that
 \begin{equation} \label{eq:uniform:bound:cts}
 G_{\nu}(\lambda)\leq G(\lambda)\leq G_{\nu}(\lambda) + \iota(\nu) \quad \text{for all }\lambda \in \Lambda,
 \end{equation}
 i.e., $G_{\nu}(\lambda)$ is a uniform approximation of the non-smooth function $G(\lambda)$. In Lemma~\ref{lem:iota:new} an explicit expression for $\iota$ is given, which implies that $\iota(\nu)\to 0$ as $\nu\to 0$.
 \begin{myass}[Lipschitz continuity] \label{a:constraint_fct:s}	\
	\begin{enumerate}[(i)]
		\item The input constraint function $s(\cdot)$ is Lipschitz continuous with constant $L_{s}$. \label{a:constraint_fct:s:i}
		\item The function $R \ni x \mapsto \rho_x \in \mathcal{D}(\mathcal{H})$ is Lipschitz continuous with constant $L$ with respect to the trace norm.\label{a:constraint_fct:s:ii}
\end{enumerate}
\end{myass}

\begin{mylem}
Assumption~\ref{a:constraint_fct:s}\eqref{a:constraint_fct:s:ii} implies that the function $f_{\lambda}(x):=\WW\lambda(x) - \Hh{\rho_x}$ for $x\in R$ is Lipschitz continuous uniformly in $\lambda \in \Lambda$ with constant $L_f:=L(M \log(\gamma^{-1} \vee \e)+\sqrt{M}\log(\tfrac{1}{\gamma \e}\vee \e))$.   
\end{mylem}
\begin{proof}
For $x_1,x_2 \in R$ using the triangle inequality we obtain
\begin{align}
\left|f_{\lambda}(x_1)-f_{\lambda}(x_2) \right| &= \left|\inprod{\rho_{x_1}}{\lambda}_F-\Hh{\rho_{x_1}}-\inprod{\rho_{x_2}}{\lambda}_F+\Hh{\rho_{x_2}} \right|\nonumber\\
&\leq \left| \inprod{\rho_{x_1}-\rho_{x_2}}{\lambda}_F \right| + \left| \Hh{\rho_{x_1}}-\Hh{\rho_{x_2}} \right|. \label{eq:nowBound}
\end{align}
We can bound the first term of \eqref{eq:nowBound} using the Cauchy-Schwarz inequality as
\begin{align}
\left| \inprod{\rho_{x_1}-\rho_{x_2}}{\lambda}_F \right| &\leq \norm{\rho_{x_1}-\rho_{x_2}}_F \norm{\lambda}_F \nonumber\\
&\leq \norm{\rho_{x_1}-\rho_{x_2}}_{\trnorm} \norm{\lambda}_F \nonumber\\
& \leq L  |x_1 - x_2| \norm{\lambda}_F, \label{eq:lasttt}
\end{align}
where \eqref{eq:lasttt} follows by Assumption~\ref{a:constraint_fct:s}\eqref{a:constraint_fct:s:ii} and by assumption $\norm{\lambda}_{F}\leq M \log\left(\gamma^{-1} \vee \e \right)$. Let $J_M:=\sqrt{M}\log(\tfrac{1}{\gamma \e}\vee \e)$, using Claim~\ref{claim:entropyLip} and Assumption~\ref{ass:channel:CQcont} the second term of \eqref{eq:nowBound} can be bounded as
\begin{align}
 \left| \Hh{\rho_{x_1}}-\Hh{\rho_{x_2}} \right| &\leq J_M \norm{\rho_{x_1}-\rho_{x_2}}_{\trnorm} \nonumber \\
  & \leq  J_M L |x_1 - x_2|, \label{eq:lasttttt}
\end{align}
where \eqref{eq:lasttttt} follows again by Assumption~\ref{a:constraint_fct:s}\eqref{a:constraint_fct:s:ii}.
\end{proof}

%Define $\underline{s}:=\min_{x\in\A}s(x)$, $\overline{s}:=\max_{x\in\A}s(x)$ and let $L_{s}$ denote the Lipschitz constant of $s(\cdot)$.
\begin{mylem}[\cite{TobiasSutter14}] \label{lem:iota:new}
Under Assumption~\ref{a:constraint_fct:s} a possible choice of the function $\iota$ in \eqref{eq:uniform:bound:cts} is given by
\begin{equation*}
\iota(\nu) = \left\{ 
  \begin{array}{l l}
    \nu \left( \log\left( \frac{T_{1}}{\nu}+T_{2} \right) +1 \right), & \quad \nu<\tfrac{T_{1}}{1-T_{2}} \text{ or } T_{2}>1\\
    \nu, & \quad \text{otherwise},
  \end{array} \right.
\end{equation*}
where $T_{1}:=L_{f}\upsilon + 2L_{f}L_{s}\upsilon^{2}\left( \tfrac{1}{-\underline{s}} \vee \tfrac{1}{\overline{s}} \right)$, $T_{2}:= L_{s}\upsilon (\underline{\mu} \, \vee \, \overline{\mu})$, $\underline{\mu}:=\tfrac{2}{-\underline{s}}\log\left( \tfrac{2L_{s}\upsilon}{-\underline{s}}\vee 1 \right)$, $\overline{\mu}:=\tfrac{2}{\overline{s}}\log\left( \tfrac{2L_{s}\upsilon}{\overline{s}}\vee 1 \right)$, $\upsilon:=\int_{R}\drv x$, $\underline{s} := -S + \min_{x\in R} s(x)$ and $\overline{s} := -S + \max_{x\in R} s(x)$.
\end{mylem}
 \begin{myremark}
In case of no input constraints, the unique optimizer to \eqref{eq:cqContGnu} is given by
\begin{equation*}
p_{\nu}^{\lambda}(x) = \frac{2^{\frac{1}{\nu}({\mathrm{tr}}[\rho_x \lambda] - \Hh{\rho_x}) }}{\int_{R} 2^{\frac{1}{\nu}({\mathrm{tr}}[\rho_x \lambda ] - \Hh{\rho_x})} \drv x},
\end{equation*}
whose straightforward evaluation is numerically difficult for small $\nu$. A numerically stable technique to evaluate the above integral for small $\nu$ can be obtained by following the method presented in Remark~\ref{remark:comp:stablility}.
\end{myremark}
\begin{myremark}[\cite{TobiasSutter14}]\label{rmk:finite:constraint:optimizer:CQ:cts}
As already highlighted and discussed in Remark~\ref{rmk:finite:constraint:optimizer:CQ}, in case of additional input constraints, we seek for an efficient method to find the coefficients $\mu_{i}$ in \eqref{eq:optimizer:CQ:cts}. Similarly to the finite input alphabet case the problem of finding $\mu_{i}$ can be reduced to the finite dimensional convex optimization problem \cite[p.~257 ff.]{ref:Lasserre-11}
\begin{equation} \label{eq:opt:problem:find:mu:cts}
\sup\limits_{\mu\in\R^{2}}\left\{ \inprod{y}{\mu} -\int_{R}p_{\nu}^{\lambda}(x)\drv x \right\},
\end{equation}
where $y:=(1,S)$. Note that $\eqref{eq:opt:problem:find:mu:cts}$ is an unconstrained maximization of a concave function. 
However, unlike to the finite input alphabet case, the evaluation of its gradient and Hessian involves computing moments of the measure $p_{\nu}^{\lambda}(x,\mu)\drv x$, which we want to avoid in view of computational efficiency. There are efficient numerical schemes known, based on semidefinite programming, to compute the gradient and Hessian (see \cite[p.~259 ff.]{ref:Lasserre-11} for details).
\end{myremark}
\begin{mylem}[{\cite[Lem.~3.14]{TobiasSutter14}}]\label{lem:strong:convexity:CQ:cts}
The function $d:\dens \to\Rp$, $p\mapsto -h(p)+ \log (\upsilon) $ with $\upsilon:=\int_{R}\drv x$ as introduced in \eqref{eq:cqContGnu} is strongly convex with convexity parameter $1$.
\end{mylem}
Finally, we can show that the uniform approximation $G_{\nu}(\lambda)$ is smooth and has a Lipschitz continuous gradient with known constant.
The following result is a generalization of Proposition~\ref{prop:CQ:lipschitz} and follows from Theorem~5.1 in \cite{ref:devolder-12}.
\begin{myprop}[Lipschitz constant of $\nabla G_{\nu}$] \label{prop:Lipschitz:continuity}
The function $G_{\nu}(\lambda)$ is well defined and continuously differentiable at any $\lambda\in\Hop$. Moreover, this function is convex and its gradient 
\begin{equation*}
\nabla G_{\nu}(\lambda) = \int_{R} \rho_{x}\,p_{\nu}^{\lambda}(x) \drv x
\end{equation*}
is Lipschitz continuous with constant
$L_{\nu} = \frac{1}{\nu} $ with respect to the Frobenius norm.
\end{myprop}
\begin{proof}
See Appendix~\ref{app:propLipCont}.
\end{proof}
We consider the smooth, convex optimization problem
\begin{align}\label{Lagrange:Dual:Program:smooth:cts:CQ}
 \mathsf{D}_{\nu}: \left\{ \begin{array}{ll}
	\min\limits_{\lambda} 		& F(\lambda) + G_{\nu}(\lambda) \\
			\st 					& \lambda\in \Lambda,
	\end{array}\right.
\end{align}
whose solution can be approximated with the Algorithm~\hyperlink{algo:1}{1} presented in Section~\ref{sec:cq:channel}. For the parameter $D_{1}:=\tfrac{1}{2}(M \log(\gamma^{-1}\vee \e))^2$ we have the following result, when running Algorithm~\hyperlink{algo:1}{1} on the problem \eqref{Lagrange:Dual:Program:smooth:cts:CQ}.
 
 \begin{mythm}\label{thm:error:bound:capacity:continuous:CQchannel}
Let $\alpha := 2(T_1 + T_2 + 1)$ where $T_1$ and $T_2$ are as defined in Lemma~\ref{lem:iota:new}. Given a precision $\varepsilon \in (0, \tfrac{\alpha}{4})$, we set the smoothing parameter $\nu = \tfrac{\varepsilon / \alpha}{\log \left( \alpha / \varepsilon \right)}$ and number of iterations $k\geq \tfrac{1}{\varepsilon} \sqrt{16 D_{1}\alpha}\sqrt{\log(\varepsilon^{-1}) + \log(\alpha) + \tfrac{1}{2}}$.  Consider 
\begin{align} 
\hat{\lambda} = y_{k} \in \Lambda \qquad \text{and} \qquad \hat{p}=\sum_{i=0}^{n}\frac{2(i+1)}{(k+1)(k+2)} p_{\nu}^{\lambda_{i}}\in \mathcal{D}(R), \label{eq:optInPut:cts}
\end{align}
where $y_i$ computed at the $\text{i}^{\text{th}}$ iteration of Algorithm~\hyperlink{algo:1}{1} and $p_{\nu}^{\lambda_{i}}$ is the analytical solution in \eqref{eq:optimizer:CQ:cts}. Then, $\hat{\lambda}$ and $\hat{p}$ are the approximate solutions to the problems  \eqref{eq:cqContCapacityDual} and \eqref{eq:cqContCapacityPrimal}, i.e., 
\begin{align}
0\leq F(\hat{\lambda}) + G(\hat{\lambda}) - \I{\hat{p}}{\rho} \leq \varepsilon. \label{eq:EBB:cts}
\end{align}
Therefore, Algorithm~\hyperlink{algo:1}{1} requires $O\left( \tfrac{1}{\varepsilon}\sqrt{\log\left( \varepsilon^{-1} \right)} \right)$ iterations to find an $\varepsilon$-solution to the problems  \eqref{eq:cqContCapacityPrimal} and \eqref{eq:cqContCapacityDual}.
\end{mythm}
\begin{proof}
The proof is a minor modification of \cite[Thm.~3.15]{TobiasSutter14}.
\end{proof}
% 
% 
% 
% 
% 
%\begin{mythm}\label{thm:error:bound:capacity:continuous:CQchannel}
%{\cts completely new theorem --- pay attention to different strong convexity parameters} Consider a smoothing parameter
%\begin{align*}
%\nu = \nu(n) = \frac{2}{n+1}\sqrt{\frac{D_{1}}{D_{2}}}.
%\end{align*}
%Then after $n$ iterations we can generate the approximate solutions to the problems \eqref{eq:cqContCapacityDual} and \eqref{eq:cqContCapacityPrimal}, namely,
%\begin{align*} 
%\hat{\lambda} = y_{n} \in \Lambda,\qquad \qquad \hat{p}=\sum_{i=0}^{n}\frac{2(i+1)}{(n+1)(n+2)} p_{\nu}^{\lambda_{i}}\in \dens,
%\end{align*}
%which satisfy the following inequality:
%\begin{align*}
%0\leq F(\hat{\lambda}) + G(\hat{\lambda}) - \I{\hat{p}}{W} \leq \frac{4}{n+1} \sqrt{D_{1} D_{2}} + \frac{4 D_{1}}{(n+1)^{2}}. 
%\end{align*}
%Thus, the complexity of finding an $\varepsilon$-solution to the problems \eqref{eq:cqContCapacityDual} and \eqref{eq:cqContCapacityPrimal} by the smoothing technique does not exceed
%\begin{align*}
%4  \sqrt{D_{1} D_{2}} \ \frac{1}{\varepsilon} + 2 \sqrt{\frac{D_{1}}{\varepsilon}}.
%\end{align*}
%\end{mythm}
%\begin{proof}
%The theorem can be proven by following the proof in \cite{nesterov05} while using Proposition~\ref{prop:Lipschitz:continuity}.
%\end{proof}
Let us highlight that we have two different quantitative bounds for the approximation error. First, the \emph{a priori} bound $\varepsilon$ for which Theorem~\ref{thm:error:bound:capacity:continuous:CQchannel} prescribes a lower bound for the required number of iterations. Second, we have an \emph{a posteriori} bound $F(\hat{\lambda}) + G(\hat{\lambda}) -  \I{\hat{p}}{\rho}$
after $k$ iterations. In practice, the a posteriori bound often approaches $\varepsilon$ within significantly less number of iterations than predicted by Theorem~\ref{thm:error:bound:capacity:continuous:CQchannel}. Besides, note that 
by \eqref{eq:uniform:bound:cts} and Theorem~\ref{thm:error:bound:capacity:continuous:CQchannel}
\begin{align*}		
0\leq F(\hat{\lambda}) + G_{\nu}(\hat{\lambda}) +\iota(\nu) -  \I{\hat{p}}{\rho} \leq \iota(\nu) + \varepsilon,
\end{align*}
which shows that $ F(\hat{\lambda}) + G_{\nu}(\hat{\lambda}) +\iota(\nu)$ is an upper bound for the channel capacity with a priori error $\iota(\nu) + \varepsilon$. This bound can be particularly helpful in cases where an evaluation of $G(\lambda)$ for a given $\lambda$ is hard.

\begin{myremark}[No input constraint]
In the absence of an input constraint we can derive an analytical expression for $G_{\nu}(\lambda)$ and its gradient. As derived above, the optimizer solving $\eqref{eq:cqContGnu}$ is 
\begin{equation*}
p^{\star}(x)=\frac{2^{{\mathrm{tr}}[\rho_x \lambda  ] - \Hh{\rho_x}}}{\int_{R}2^{{\mathrm{tr}}[\rho_y \lambda ] -\Hh{\rho_y}} \drv y}, \quad x \in R, 
\end{equation*}
which gives
\begin{equation*}
G_{\nu}(\lambda)= \nu \log \left(\int_{R} 2^{\tfrac{1}{\nu}\left({\mathrm{tr}}[\rho_x \lambda] -\Hh{\rho_x} \right)} \drv x\right)-\nu \log\left(\upsilon \right)
\end{equation*}
and
\begin{equation} \label{eq:gradient:G:no:input:constr}
\frac{\partial G_{\nu}(\lambda)}{\partial \lambda_{m,\ell}} = \left(\nabla G_{\nu}(\lambda) \right)_{m,\ell}= \frac{1}{S(\lambda)} \int_{R} 2^{\tfrac{1}{\nu}\left({\mathrm{tr}}[\rho_x \lambda ] -\Hh{\rho_x}\right)} (\rho_x)_{\ell,m} \drv x,
\end{equation}
with $S(\lambda)=\int_{R} 2^{\tfrac{1}{\nu}\left({\mathrm{tr}}[ \rho_x \lambda  ]-\Hh{\rho_x}\right)}\drv x$. Similarly to Remark~\ref{remark:comp:stablility}, we have used $\tfrac{\partial \Tr{\rho \lambda }}{\partial \lambda_{m,\ell}}=\rho_{\ell,m}$ \cite[Prop.~10.7.2]{ref:Ber-09}.
\end{myremark}
\subsection{Inexact first-order information} \label{sec:subsection:inexact}
Our analysis up to now assumes availability of exact first-order information, namely we assumed that the gradients $\nabla G_{\nu}(\lambda)$ and $\nabla F(\lambda)$ are exactly available for any $\lambda$. However, in many cases, e.g., in the presence of an additional input cost constraint (Remark~\ref{rmk:finite:constraint:optimizer:CQ:cts}), the evaluation of those gradients requires solving another auxiliary optimization problem or a multi-dimensional integral \eqref{eq:gradient:G:no:input:constr}, which only can be done approximately. This motivates the question of how to solve \eqref{Lagrange:Dual:Program:smooth:cts:CQ} in the case of inexact first-order information which indeed has been studied in detail in \cite{ref:Devolver-13}. In our problem \eqref{Lagrange:Dual:Program:smooth:cts:CQ}, $\nabla F(\lambda)$ has a closed form expression \eqref{eq:grad:F} and as such can be assumed to be known exactly. Let us assume, however, that we only have an oracle providing an approximation $\nabla \tilde{G}_{\nu}(\lambda)$, which satisfies $\norm{\nabla \tilde{G}_{\nu}(\lambda)-\nabla G_{\nu}(\lambda)}_{\opnorm}\leq \delta$ for any $\lambda\in\Lambda$ and some $\delta>0$. Recall that $\pi_{\Lambda}$, as defined in Proposition~\ref{prop:proj}, denotes the projection operator onto the set $\Lambda$, defined in Lemma~\ref{lem:compact:set:CQ:cts}, that is the Frobenius norm ball with radius $r:=M\log\left(\gamma^{-1} \vee \e \right)$.
\begin{table}[!htb]
\centering 
\begin{tabular}{c}
  \Xhline{3\arrayrulewidth}  \hspace{1mm} \vspace{-3mm}\\ 
\hspace{35mm}{\bf{\hypertarget{algo:2}{Algorithm 2: } }} Scheme for inexact first-order information \hspace{34mm} \\ \vspace{-3mm} \\ \hline \vspace{-0.5mm}
\end{tabular} \\
\vspace{-5mm}

 \begin{flushleft}
  {\hspace{1.7mm}Choose some $\lambda_0 \in \Hop$}
 \end{flushleft}
 \vspace{-8mm}

 \begin{flushleft}
  {\hspace{1.7mm}\bf{For $m\geq 0$ do$^{*}$}}
 \end{flushleft}
 \vspace{-8mm}
 
  \begin{tabular}{l l}
{\bf Step 1: } & Compute $\nabla F(\lambda_{m})+\nabla \tilde{G}_{\nu}(\lambda_{m})$ \\
{\bf Step 2: } & $\lambda_{m+1} = \pi_{\Lambda}\left(-\frac{1}{L_{\nu}}\left( \nabla F(\lambda_m)+\nabla \tilde{G}_{\nu}(\lambda_m) \right) + \lambda_m\right)$\\
\vspace{-7mm}
  \end{tabular}
  
  \begin{flushleft}
  {\hspace{1.7mm}[*The stopping criterion is explained in Remark~\ref{remark:stopping2}]}
  \vspace{-10mm}
 \end{flushleft}   
  
\begin{tabular}{c}
\hspace{38mm} \phantom{ {\bf{Algorithm:}} Scheme for inexact first-order information}\hspace{36.5mm} \\ \vspace{-1.0mm} \\\Xhline{3\arrayrulewidth}
\end{tabular}
\end{table}
\begin{myprop}\label{prop:inexact:oracle}
For every $\nu \in \Rps$, after $k$ iterations of Algorithm~\hyperlink{algo:2}{2}
\begin{equation}
 F(\lambda_{k})+G(\lambda_{k})-C_{\mathsf{cq},S}(\W)  \leq \frac{(2+\tfrac{1}{\nu})D^{2}}{2k}+\iota(\nu) + 2\delta D, \label{eq:robustError}
\end{equation}
where $\iota(\nu)$ is given in Lemma~\ref{lem:iota:new} and $D:= M\log\left(\gamma^{-1} \vee \e \right)$.
\end{myprop}
\begin{proof}
We denote the optimum value to \eqref{Lagrange:Dual:Program:smooth:cts:CQ} by $C_{\nu,\mathsf{cq},S}(\W)$. According to \cite{ref:Devolver-13}, for every $\nu \in \Rps$, after $k$ iterations of Algorithm~\hyperlink{algo:2}{2}
\begin{equation}\label{eq:proof:inexact:devolder}
 F(\lambda_{k})+G(\lambda_{k})-C_{\nu,\mathsf{cq},S}(\W)  \leq \frac{(2+\tfrac{1}{\nu})D^{2}}{2k} + 2\delta D.
\end{equation}
By recalling $\eqref{eq:uniform:bound:cts}$, which leads to $C_{\nu,\mathsf{cq},S}(\W)\leq C_{\mathsf{cq},S}(\W)$ the statement \label{eq:proof:inexact:devolder} can be refined to
\begin{equation*}
 F(\lambda_{k})+G(\lambda_{k})-C_{\mathsf{cq},S}(\W)  \leq \frac{(2+\tfrac{1}{\nu})D^{2}}{2k}+\iota(\nu) + 2\delta D.
\end{equation*}
\end{proof}

\begin{myremark}[Stopping criterion of Algorithm~\hyperlink{algo:1}{2}] \label{remark:stopping2}
In case of no average power constraint the following explicit formulas can be used as a stopping criterion of Algorithm~\hyperlink{algo:1}{2}. Choose an a priori error $\varepsilon>0$. For $\beta:=1+\tfrac{\log \e}{\e}$ and $\alpha:=\log T_{1} +1$, where $T_{1}$ is as in Lemma~\ref{lem:iota:new}, consider $\nu\leq \tfrac{\varepsilon}{3 \beta \left( \alpha + \log(3\beta \varepsilon^{-1})\right)}$, $k\geq \tfrac{3(M\log(\gamma^{-1} \vee \e ))^{2}(2\varepsilon + 3 \beta ( \alpha + \log(3\beta\varepsilon^{-1})))}{2\varepsilon^{2}}$ and $\delta\leq \tfrac{\varepsilon}{6M\log\left(\gamma^{-1} \vee \e \right)}$. For this choice Algorithm~\hyperlink{algo:1}{2} guarantees an $\varepsilon$-close solution, i.e., the right hand side of \eqref{eq:robustError} is upper bounded by $\varepsilon$. This analysis follows by Lemma~\ref{lem:lemma1:stop:proof} that is given in Appendix~\ref{app:lemma1}.
\end{myremark}

%%%%%%%%%%%%%%%%%%%%%%%%%%%%%%%%%
%%%%%%%%%%%%%%%%%%%%%%%%%%%%%%%%%
\section{Approximating the Holevo Capacity} \label{sec:Holevo}
In this section it is shown how ideas developed in the previous sections for cq channels can be extended to quantum channels with a quantum mechanical input and output, also known as qq channels.
Let $\Phi:\mathcal{B}(\mathcal{H}_A)\to \mathcal{B}(\mathcal{H}_B)$ be a quantum channel, where $\mathcal{B}(\mathcal{H})$ denotes the space of bounded linear operators on some Hilbert space $\mathcal{H}$ that are equipped with the trace norm. The classical capacity describing the maximal amount of classical information that can be sent on average, asymptotically reliable over the channel $\Phi$ per channel use,  has proven to be \cite{holevo98,schumacher97}
\begin{equation}
C(\Phi)= \lim_{k\to \infty} \frac{1}{k} C_{\mathcal{X}}(\Phi^{\otimes k}), \label{eq:clCapacity} 
\end{equation}
where
\begin{equation}
C_{\mathcal{X}}(\Phi)= \sup \limits_{\{p_i,\rho_i\}} \Hh{\sum_i p_i \Phi(\rho_i)} - \sum_i p_i \Hh{\Phi(\rho_i)}, \label{eq:holevoCapacity}
\end{equation}
denotes the Holevo capacity. It is immediate to verify that $C(\Phi) \geq C_{\mathcal{X}}(\Phi)$ for all quantum channels $\Phi$. In 2008, the existence of channels satisfying $C(\Phi) > C_{\mathcal{X}}(\Phi)$ has been proven which implies that the limit in \eqref{eq:clCapacity} which is called \emph{regularization} is necessary \cite{hastings09}. Due to the regularization, a direct approximation of $C(\Phi)$ seems difficult. 

In this section, we present an approximation scheme for the Holevo capacity based on the method explained in Section~\ref{sec:cq:contInput}. It has been shown that the supremum in \eqref{eq:holevoCapacity} is attained on an ensemble consisting of no more than $N^2$ pure states, where $N:=\dim \mathcal{H}_A$ \cite[Cor.~8.5]{holevo_book}. The Holevo capacity is in general hard to compute since \eqref{eq:holevoCapacity} is a non-convex optimization problem as the objective function is concave in $\{p_i\}$ for fixed $\{\rho_i\}$ and convex in $\{\rho_i\}$ for fixed $\{p_i\}$ \cite[Thms.~12.3.5 and 12.3.6]{wilde_book}. Furthermore, Beigi and Shor showed that computing the Holevo capacity is $\mathsf{NP}$-complete \cite{shor08}. Their proof also implies that it is $\mathsf{NP}$-hard to compute the Holevo capacity up to $\tfrac{1}{\poly(N)}$ accuracy. Based on a stronger complexity assumption, Harrow and Montanaro improved this result by showing that the Holevo capacity is in general hard to approximate even up to a constant accuracy \cite{harrow13}.

Using a universal encoder, which is a mapping translating a classical state into a quantum state, we can compute the Holevo capacity of a quantum channel by calculating the cq capacity of a channel having a continuous, bounded input alphabet (see Figure~\ref{fig:cqHolevo}). A universal encoder is defined as the mapping $ \Eu: R \ni r \mapsto \ket{r}\! \bra{r}=:\rho_r \in \mathcal{D}(\mathcal{H}_A)$.  
From an optimization point of view, by adding the universal encoder we map a finite dimensional non-convex optimization problem (of the form \eqref{eq:holevoCapacity}) into an infinite dimensional convex optimization problem (of the form \eqref{eq:cqContCapacity}), which we know how to approximate as discussed in Section~\ref{sec:cq:contInput}.
To represent an $N$ dimensional pure state we need $2N-2$ real bounded variables.\footnote{We need to describe an $N$ dimensional complex vector, where one real parameter can be removed since the global phase is irrelevant. A second parameter is determined as the vector must have unit length.} As an example, for $N=2$ a possible universal encoder is $\Eu: [0,\pi] \times [0,2\pi]  \owns (\phi,\theta) \mapsto \ket{v}\! \bra{v} \in \mathbb{C}^{2 \times 2}$, with $\ket{v}=(\cos \theta,\sin \theta \, e^{\ii \phi})\transp$. A possible universal encoder for a general $N$ dimensional setup is discussed in Remark~\ref{rmk:universalEncoder}.
\begin{figure}[!htb]
\centering
\def \xenc{1.0}
\def \yenc{0.4}
\def \xchannel{1.0}
\def \ychannel{0.4}

\def \x{0.8}
\def \xm{1.0} %distance in the middle
\def \la{0.2}

\def \xdash{0.25}
\def \ydash{0.25}

\begin{tikzpicture}[scale=1,auto, node distance=1cm,>=latex']
	
 \draw[->] (0,0) -- (\x,0);   
 \node at (-1*\la,0) {$r$};  
%encoder
 \draw[] (\x,-\yenc) -- (\x,\yenc);   
 \draw[] (\x+\xenc,-\yenc) -- (\x+\xenc,\yenc);   
 \draw[] (\x,-\yenc) -- (\x+\xenc,-\yenc);   
 \draw[] (\x,\yenc) -- (\x+\xenc,\yenc);   
 \node at (\x+0.5*\xenc,0) {$\Eu$};       

 \draw[->] (\x+\xenc,0) -- (\x+\xm+\xenc,0);   
 \node at (\x+0.5*\xm+\xenc,1.5*\la) {$\rho_r$};     

%channel
  \draw[] (\x+\xm+\xenc,-\ychannel) -- (\x+\xm+\xenc,\ychannel);   
  \draw[] (\x+\xm+\xenc+\xchannel,-\ychannel) -- (\x+\xm+\xenc+\xchannel,\ychannel);    
  \draw[] (\x+\xm+\xenc,-\ychannel) -- (\x+\xm+\xenc+\xchannel,-\ychannel);    
  \draw[] (\x+\xm+\xenc,\ychannel) -- (\x+\xm+\xenc+\xchannel,\ychannel);  
  \node at (\x+\xm+\xenc+0.5*\xchannel,0) {$\Phi$};   
    
 \draw[->] (\x+\xm+\xenc+\xchannel,0) -- (2*\x+\xm+\xenc+\xchannel,0);   
 \node at (2*\x+\xm+\xenc+\xchannel+2*\la,0) {$\sigma_r$};

 % dashed box
   \draw[dashed] (\x-\xdash,\yenc+\ydash) -- (\x+\xm+\xenc+\xchannel+\xdash,\yenc+\ydash);   
   \draw[dashed] (\x-\xdash,-\yenc-\ydash) -- (\x+\xm+\xenc+\xchannel+\xdash,-\yenc-\ydash);   
   \draw[dashed] (\x-\xdash,\yenc+\ydash) -- (\x-\xdash,-\yenc-\ydash);   
   \draw[dashed] (\x+\xm+\xenc+\xchannel+\xdash,-\yenc-\ydash) -- (\x+\xm+\xenc+\xchannel+\xdash,\yenc+\ydash);   
   \node at (\x+0.5*\xm+0.5*\xenc+0.5*\xchannel,\yenc+\ydash+\la) {$\W$};            
\end{tikzpicture}
\caption{\small Using a universal encoder $ \Eu: R \ni r \mapsto \ket{r}\! \bra{r}=:\rho_r \in \mathcal{D}(\mathcal{H}_A)$ we embed the quantum channel $\Phi:\mathcal{B}(\mathcal{H}_A) \to \mathcal{B}(\mathcal{H}_B)$ into a cq channel $\W: R \ni r \mapsto (\Phi\circ \Eu)(r)= \Phi( \ket{r}\! \bra{r})=:\sigma_r \in \mathcal{D}(\mathcal{H}_B)$ having a continuous bounded input alphabet. We then have $C_{\mathsf{cq}}(\W)=C_{\mathcal{X}}(\Phi)$, with $C_{\mathsf{cq}}(\W)$ and $C_{\mathcal{X}}(\Phi)$ as defined in \eqref{eq:cqCapacity} and \eqref{eq:holevoCapacity}.}
\label{fig:cqHolevo}
\end{figure}

As explained in Figure~\ref{fig:cqHolevo}, using the idea of the universal encoder gives $C_{\mathsf{cq}}(\W)=C_{\mathcal{X}}(\Phi)$, i.e., we can approximate $C_{\mathcal{X}}(\Phi)$ by approximating $C_{\mathsf{cq}}(\W)$. This can be done as explained in Section~\ref{sec:cq:contInput}. 
For an approximation error $\varepsilon >0$, Theorem~\ref{thm:error:bound:capacity:continuous:CQchannel} gives a minimal number of iterations $k$ and a smoothing parameter $\nu >0$ such that after $k$ iterations Algorithm~\hyperlink{algo:1}{1} generates a lower and upper bound $C_{\mathcal{X},\textnormal{LB}}(\Phi)\leq C_{\mathcal{X}}(\Phi) \leq C_{\mathcal{X},\textnormal{UB}}(\Phi)$ to the Holevo capacity such that 
\begin{equation*}
0 \leq  C_{\mathcal{X},\textnormal{UB}}(\Phi) - C_{\mathcal{X},\textnormal{LB}}(\Phi) \leq \varepsilon. 
\end{equation*}

\begin{myremark}[Universal encoder] \label{rmk:universalEncoder}
For a channel $\Phi:\mathcal{B}(\mathcal{H}_A)\to \mathcal{B}(\mathcal{H}_B)$ with $N=\dim \mathcal{H}_A$ a possible universal encoder can be derived using spherical coordinates as
\begin{align*}
&\Eu: R = [0,\pi]\times \ldots \times [0,\pi] \times [0,2\pi] \times [0,\pi] \times \ldots \times [0,\pi] \to \mathbb{C}^{N \times N }\\
&\hspace{5mm}(\theta_1,\ldots,\theta_{N-2},\theta_{N-1},\phi_1,\ldots,\phi_{N-1}) \mapsto  \ket{v}\! \bra{v}
\end{align*}
with
\begin{align*}
\ket{v}&=(\cos \theta_1,\sin \theta_1 \cos \theta_2 e^{\ii \phi_1},\sin \theta_1 \sin \theta_2 \cos \theta_3 e^{\ii \phi_2},\ldots,\sin \theta_1 \ldots \sin \theta_{N-2} \cos \theta_{N-1} e^{\ii \phi_{N-2}},\\
&\hspace{8mm}\sin \theta_1 \ldots \sin \theta_{N-2} \sin \theta_{N-1} e^{\ii \phi_{N-1}})\transp.
\end{align*}
It can be verified immediately that the Lebesgue measure of the set $R$ is equal to $2\pi^{2 N-2}$ for this setup.
\end{myremark}

%%%%%%%%%%%%%%%%%%%%%%%%%%%%%%%%%%%%%
\subsection{Computational complexity}
Let $\{\Phi:\mathcal{B}(\mathcal{H}_A)\to \mathcal{B}(\mathcal{H}_B)\}_{N,M}$ be a family of quantum channels with $N:=\dim \mathcal{H}_A$ and $M:=\dim \mathcal{H}_B$. For such a family, we derive the complexity of our method presented in this chapter to ensure an $\varepsilon$-close solution. 
Suppose the family of channels $\{\Phi\}_{N,M}$ satisfies the following assumption.
\begin{myass}[Regularity] \label{ass:FPRAS}
$\gamma_M:=\min \limits_{\rho \in \mathcal{D}(\mathcal{H}_A)} \min \spec \Phi(\rho)>0$ 
\end{myass}
To simplify notation, define the function $\Rp\ni M\mapsto \p(M):=\log\left( \gamma_M^{-1} \right) \in \Rp$. 
We will discuss later in Remark~\ref{rmk:avoidingAssFPRAS} how Assumption~\ref{ass:FPRAS} can be removed at the cost of computational complexity proportional to $\varepsilon^{-1} \log \varepsilon^{-1}$ where $\varepsilon$ is the preassigned approximation error, i.e.,  considering $\varepsilon$ as a constant Assumption~\ref{ass:FPRAS} can be automatically satisfied. As detailed in the preceding section and summarized in Algorithm~\hyperlink{algo:1}{1}, for the approximation of the Holevo capacity one requires to efficiently evaluate the gradient $\nabla G_{\nu}(\lambda)$ for an arbitrary $\lambda \in \Lambda$ given by \eqref{eq:gradient:G:no:input:constr}, which involves two integrations over $R$. 

\begin{mydef}[Gradient oracle complexity] \label{def:oracle}
Given a family of channels $\{\Phi\}_{N,M}$, the computational complexity for Algorithm $\mathscr{A}$ to provide an estimate $\tilde{G}_{\nu}(\lambda)$ for any $\lambda\in\Lambda$ of the form
\begin{equation*}
\Prob{\norm{\nabla G_{\nu}(\lambda) - \nabla \tilde{G}_{\nu}(\lambda)}_{\opnorm}\geq \delta} \leq \eta
\end{equation*}
is denoted (when it exists) by $\mathscr{C}_{\Phi,\mathscr{A}}(N,M,\delta^{-1}, \eta^{-1})$.\footnote{Note that $\mathscr{C}_{\Phi,\mathscr{A}}(N,M,\delta^{-1}, \eta^{-1})$ is increasing in all its components.}
\end{mydef}
In Sections~\ref{subsec:algoForGradient} and \ref{subsec:algoForGradient:second}, we discuss two candidates for $\mathscr{A}$ and derive their complexity as defined in Definition~\ref{def:oracle}.

\begin{mythm}[Complexity of Algorithm~\hyperlink{algo:2}{2}] \label{thm:FRRAS}
Let $\{\Phi\}_{N,M}$ be a family of quantum channels satisfying Assumption~\ref{ass:FPRAS}. Then, Algorithm~\hyperlink{algo:2}{2} together with $\mathscr{A}$ require 
\begin{align*}
&O\Big(\varepsilon^{-2} M^4 \p(M)^2\big(N +\log(M \p(M)) + \log(\varepsilon^{-1})\big)  \\ 
&\hspace{7mm} \mathscr{C}_{\Phi,\mathscr{A}}\big(N,M,\varepsilon^{-1} M \p(M), \xi^{-1} \varepsilon^{-2} M^2 \p(M)^2 (N + \log(M \p(M))+\log(\varepsilon^{-1}))\big)\Big)
\end{align*}
to compute an $\varepsilon$-close solution to the Holevo capacity with probability $1-\xi$. 
\end{mythm}
\begin{myremark}
Theorem~\ref{thm:FRRAS} establishes a link in terms of computational complexity from the main objective of this section, the Holevo capacity of a family of quantum channels $\{\Phi\}_{N,M}$ under Assumption~\ref{ass:FPRAS}, to the computation of $\nabla G_\nu(\lambda)$ for a given $\lambda\in \Lambda$, the task of Algorithm $\mathscr{A}$ in Definition~\ref{def:oracle}. That is, if $\mathscr{C}(N,M,\delta^{-1},\eta^{-1})$ for $\delta^{-1}$ and $\eta^{-1}$ given in Theorem~\ref{thm:FRRAS} is polynomial (resp. sub-exponential) in $(N,\varepsilon^{-1})$, then the complexity of the proposed scheme to approximate the Holevo capacity is polynomial (resp. sub-exponential) in $(N,\varepsilon^{-1})$.
\end{myremark}

To prove Theorem~\ref{thm:FRRAS} one requires a few preparatory lemmas. First we need an explicit a priori error bound in a similar fashion as in Section~\ref{sec:cq:contInput} given that the function $f_{\lambda,M}(x):= \Tr{\Phi(\Eu(x))\lambda}-\Hh{\Phi(\Eu(x))}$ is Lipschitz continuous uniformly in $\lambda\in \Lambda$. The following lemma shows that this readily follows from Assumption~\ref{ass:FPRAS}.
\begin{mylem}\label{lem:assLip}
Let $\{\Phi\}_{N,M}$ be a family of channels satisfying Assumption~\ref{ass:FPRAS}. The function $f_{\lambda,M}(x):= \Tr{\Phi(\Eu(x))\lambda}-\Hh{\Phi(\Eu(x))}$ for $x \in R$ is Lipschitz continuous uniformly in $\lambda\in \Lambda$ with respect to the $\ell^1$-norm with constant $L_{N,M}=2 N \sqrt{N} \left( M \log(\tfrac{1}{\gamma_M} \vee \e) +  \sqrt{M} \log(\tfrac{1}{\gamma_M \e} \vee \e)\right)$.
\end{mylem}
\begin{proof}
See Appendix~\ref{app:lem:assLip}.
\end{proof}

%%%%%%%%%%%%%%%%%%%
\begin{mylem}\label{lem:Omar}
Let $\eta \in [0,1]$ and $n\in \mathbb{N}$. Then $1-(1-\eta)^n \leq n \eta$.
\end{mylem}
\begin{proof}
For a fixed $n \in \mathbb{N}$ the function $[0,1] \ni \eta \mapsto f(\eta):=1-(1-\eta)^n-n\eta$ is concave since $\tfrac{\drv^2 f(\eta)}{\drv \eta} =- n(n-1)(1-\eta)^{n-2} \leq 0$. Solving $\tfrac{\drv f(\eta)}{\drv \eta}=0$, gives $\eta^{\star}=0$. As $f(0)=0$ and $f(1)=1-n \leq 0$ this proves that $f(\eta)\leq 0$ for all $n \in \mathbb{N}$ and $\eta \in [0,1]$.
\end{proof}
%%%%%%%%%%%%%%%%%%%

%%%%%%%%%%%%%%%%%%%
\begin{proof}[Proof of Theorem~\ref{thm:FRRAS}]
Recall that according to Proposition~\ref{prop:inexact:oracle}, after $k$ iterations of Algorithm~\hyperlink{algo:2}{2}, where the gradient $\nabla G_{\nu}(\lambda_{i})$ in each iteration $i$ is approximated with $\nabla \tilde{G}_{\nu}(\lambda_{i})$ using Algorithm $\mathscr{A}$ as introduced in Definition~\ref{def:oracle}, we get
\begin{equation}\label{eq:diamond}
F(\lambda_{k})+G(\lambda_{k})-C_{\mathcal{X}}(\Phi)\leq \frac{(2+\tfrac{1}{\nu})D^{2}}{2k}+\iota(\nu) + 2 \delta D,
\end{equation}
where the function $\iota(\cdot)$ is given in \eqref{eq:iottaeq}.

%%%%%%%%%%%%%%%%%%%
As ensured by Definition~\ref{def:oracle} with probability $1-\eta$ the numerically evaluated gradient $\nabla \tilde{G}_{\nu}(\lambda)$ is close to its exact value $\nabla G_{\nu}(\lambda)$ or more precisely with probability at least $1-\eta$, $\nabla \tilde{G}_{\nu}(\lambda) \in \mathcal{A}$, where $\mathcal{A}:=\{X \in \mathbb{C}^{n \times n}: \norm{\nabla G_{\nu}(\lambda) - X}_{\opnorm}<  \delta \}$ denotes a confidence region. We first derive the complexity of finding an $\varepsilon$-close solution to $C_{\mathcal{X}}(\Phi)$ given that in every iteration step the numerically evaluated gradient lies in the confidence region $\mathcal{A}$. Afterwards we justify that the probability that the gradient in all iteration steps is evaluated approximately correctly, i.e., such that its value lies inside the confidence region, is high. 

%Lemma~\ref{lem:McDiamand} and \eqref{eq:diamond} tell us how to choose the parameters $\nu$, $\delta$, $n$ and $k$ to achieve an $\varepsilon$-close solution. 
Recall that for our setup the function $\iota(\cdot)$ in \eqref{eq:diamond} has the form 
\begin{align} \label{eq:iottaeq}
\iota(\nu)= \left\{ \begin{array}{ll}
\nu \log\left( \frac{L_{N,M}2\pi^{2N-2}}{\nu}\right)+\nu, \qquad &\quad\nu<L_{N,M}(2\pi^{2N-2}) \\
\nu, & \quad \text{otherwise},
\end{array}\right.
\end{align}
as given in Lemma~\ref{lem:iota:new} with $L_{N,M}$ defined in Lemma~\ref{lem:assLip}. Note that we use a universal encoder as introduced in Remark~\ref{rmk:universalEncoder} which gives $\upsilon=\int_R \drv x = 2 \pi^{2N-2}$. 

According to Remark~\ref{remark:stopping2} and \eqref{eq:iottaeq} we define $\beta=1+\tfrac{\log \e}{\e}$ and $\alpha:=\log(L_{N,M}) + (2N-2)\log(2 \pi) +1$, which by Lemma~\ref{lem:assLip} scales as $\alpha = O(N + \log(N^{3/2}M \p(M)))$. Following Remark~\ref{remark:stopping2} the number of iterations $k$ and the gradient approximation accuracy $\delta$ are chosen such that
\begin{align}
k&=O\Big(\varepsilon^{-2} M^2 \p(M)^2 \left(N+\log(M \p(M))+ \log(\varepsilon^{-1})\right)\Big). \label{eq:iterationsFRPAS} \\
\delta&\leq \tfrac{\varepsilon}{6M\log\left(\gamma^{-1} \vee \e \right)}=O\left(\frac{\varepsilon}{M \p(M)}\right). \label{eq:deltaaa}
\end{align}
As shown in Remark~\ref{remark:stopping2}, for these two parameters with a smoothing parameter $\nu\leq \tfrac{\varepsilon}{3 \beta \left( \alpha + \log(3\beta \varepsilon^{-1})\right)}$ after $k$ iterations of Algorithm~\hyperlink{algo:1}{2} we obtain an $\varepsilon$-close solution. The total complexity for an $\varepsilon$-solution is $k$ times the complexity of a single iteration which is 
\begin{align*}
&k \, O\!\left(M^2 \mathscr{C}_{\Phi,\mathscr{A}}(N,M,\delta^{-1}, \eta^{-1})\right) \\
&\hspace{10mm}= k\, O\!\left(M^2 \mathscr{C}_{\Phi,\mathscr{A}}(N,M,\varepsilon^{-1} M \log(\p(M)), \eta^{-1})\right) \\
&\hspace{10mm}=O\!\left(\varepsilon^{-2} M^4 \p(M)^2 \left(N+\log(M \p(M))+ \log(\varepsilon^{-1})\right)\mathscr{C}_{\Phi,\mathscr{A}}(N,M,\varepsilon^{-1} M \log(\p(M)), \eta^{-1})\right),
\end{align*}
where we used \eqref{eq:iterationsFRPAS} and \eqref{eq:deltaaa}. 

We next show that the randomized scheme is reliable with probability $1-\xi$. As mentioned in Definition~\ref{def:oracle} each evaluation of the gradient $\nabla \tilde G_{\nu}(\lambda)$ is confident with a probability not smaller than $(1-\eta)$. The scheme is successful if the gradient evaluation lies inside the confidence region in each iteration step. Thus the probability that the approximation scheme fails can be bounded by
\begin{equation*}
\Prob{\textnormal{scheme fails}} \leq 1-(1-\eta)^k \leq k \eta = O( \varepsilon^{-2} M^2 \p(M)^2 \left(N+\log(M \p(M))+ \log(\varepsilon^{-1})\right) \eta),
\end{equation*}
where the second inequality is due to Lemma~\ref{lem:Omar} and $\eqref{eq:iterationsFRPAS}$. Therefore for $\eta^{-1} = k \, \xi^{-1}$ the scheme is reliable with probability $1-\xi$.

\end{proof}

\begin{myprop}[{Continuity of the Holevo capacity \cite[Cor.~11]{leung09}}] \label{prop:continuityHolevo}
Let $\Phi_1,\Phi_2 : \mathcal{B}(\mathcal{H}_A) \to \mathcal{B}(\mathcal{H}_B)$ be two quantum channels with $M=\dim \mathcal{H}_B$ such that $\norm{\Phi_1 - \Phi_2}_{\diamond}\leq \varepsilon$ for $\varepsilon \geq 0$, then
\begin{equation*}
\left| C_{\mathcal{X}}(\Phi_1)-C_{\mathcal{X}}(\Phi_2) \right| \leq 8 \varepsilon \log(M) + 4 \Hb(\varepsilon).
\end{equation*}
\end{myprop}

\begin{myremark}[Removing Assumption~\ref{ass:FPRAS}] \label{rmk:avoidingAssFPRAS}
The continuity of the Holevo capacity can be used to remove Assumption~\ref{ass:FPRAS}. 
Let $\{\Phi_1\}_{N,M}$ be a familiy of quantum channels that violates Assumption~\ref{ass:FPRAS}. Consider the family $\{ \Phi_2\}_{N,M}:=\{(1-\xi_{N,M}) \Phi_1 + \xi_{N,M} \Theta \}_{N,M}$ for $\xi_{N,M} \in (0,1)$ with $\Theta(\rho)=\Tr{\rho} \tfrac{\mathbf{1}}{M}$. Using the triangle inequality we find for each member of the two families
\begin{equation}
\norm{\Phi_1-\Phi_2}_{\diamond} = \norm{\xi_{N,M} (\Theta-\Phi_1)}_{\diamond} \leq \xi_{N,M} (\norm{\Theta}_{\diamond}+\norm{\Phi_1}_{\diamond}) \leq 2 \xi_{N,M}, \label{eq:normDiamond}
\end{equation}
where the final inequality uses the fact that the trace norm of a channel is always upper bounded by one.
Note that the family $\{\Phi_2 \}$ as defined above clearly satisfies Assumption~\ref{ass:FPRAS} as $\Phi_2(\rho) \geq \xi_{N,M} \tfrac{\mathbf{1}}{M}$ for all $\rho \in \mathcal{D}(\mathcal{H}_A)$. This argument shows that Assumption~\ref{ass:FPRAS} is not restrictive in the sense that if one encounters a family of channels which does not satisfy it there exists another family that is close in terms of diamond norm which satisfies Assumption~\ref{ass:FPRAS} and whose Holevo capacity is very close as ensured by Proposition~\ref{prop:continuityHolevo}.
\end{myremark}

%%%%%%%%%%%%%%%%%%%%%%%%%%%%%%%%%%%%%%
\subsection{Gradient approximation} \label{subsec:algoForGradient}
As shown in the previous section, the crucial element for our approximation method is Algorithm $\mathscr{A}$ to approximate the gradient $G_{\nu}(\lambda)$ that is given in \eqref{eq:gradient:G:no:input:constr}. In this section we propose two candidates and discuss their corresponding complexity function $\mathscr{C}_{\Phi,\mathscr{A}}$.
The main idea is to approximate $\nabla G_{\nu}(\lambda)$ via a probabilistic method. 
\subsubsection*{First approach: uniform sampling}
This approach relies on a simple randomized algorithm generating independent samples from a uniform distribution. Consider 
\begin{equation} \label{eq:stochastic:approximation:of:dG}
 \nabla \tilde{G}_{\nu}(\lambda):= \frac{\sum_{i=1}^{n} 2^{\tfrac{1}{\nu}\left({\mathrm{tr}}[\Phi(\Eu(X_i)) \lambda ] -H(\Phi(\Eu(X_i)))\right)} \Phi(\Eu(X_i))}{\sum_{i=1}^{n} 2^{\tfrac{1}{\nu}\left({\mathrm{tr}}[ \Phi(\Eu(X_i)) \lambda  ]-H(\Phi(\Eu(X_i)))\right)}},
\end{equation}
where $\{X_{i}\}_{i=1}^{n}$ are i.i.d. random variables uniformly distributed on $R$.
In Lemma~\ref{lem:McDiamand} we derive a measure concentration bound to quantify the approximation error. 
As above, we denote by $L_{N,M}$ the Lipschitz constant of the function $f_{\lambda,M}(x):= \Tr{\Phi(\Eu(x))\lambda}-\Hh{\Phi(\Eu(x))}$ with respect to the $\ell^1$-norm.
\begin{mylem}\label{lem:McDiamand}
For every $0\leq \delta \leq 2^{\frac{-\sqrt{N}L_{N,M}}{\nu}-1}$
 \begin{equation*}
 \Prob{\norm{\nabla G_{\nu}(\lambda) - \nabla \tilde{G}_{\nu}(\lambda)}_{\opnorm}\geq \delta} \leq M\exp \left( -\delta^{2}nK_{N,M} \right) =: \eta
 \end{equation*}
for $K_{N,M}:=\tfrac{1}{576} 2^{\frac{-4\sqrt{N}L_{N,M}}{\nu}}$.
\end{mylem}
\begin{proof}
See Appendix~\ref{app:lem:mcdiarmid}.
\end{proof}

\begin{mycor}\label{cor:complexity}
Given a family of channels $\{\Phi\}_{N,M}$ and some fixed $\varepsilon >0$. With high probability, using Algorithm~\hyperlink{algo:2}{2} with a uniform sampling method as explained in Lemma~\ref{lem:McDiamand}, the complexity for an $\varepsilon$-close solution to the Holevo capacity is 
\begin{equation*}
O\left(M^6 \log(M \log M)^4\Big(N+\log\big(M \log(M \log M)\big)\Big) 2^{c (N^{3/2}+N^{1/2}\log(N^{3/2}M \log(M \log M)))L_{N,M}}\right),
\end{equation*}
where $c>0$ is a constant.
\end{mycor}
\begin{proof}
See Appendix~\ref{app:cor:complexity}.
\end{proof}

\begin{mycor} [Subexponential or polynomial running time] 
Let $\varepsilon >0$. Given a family of channels $\{\Phi\}_{N,M}$ with $M= \poly(N)$ such that 
\begin{enumerate}[(i)]
\item $\tfrac{1}{L_{N,M}}=\Omega(N^{3/2})$. Then the method described in this section, using an integration method explained in Lemma~\ref{lem:McDiamand}, provides with high probability an $\varepsilon$-approximation to the Holevo capacity with a complexity $O\big(M^6 \log(M \log M)^4(N+\log(M \log(M \log M))) \big) = \poly(N)$.
\item $\tfrac{1}{L_{N,M}}=\Omega(N^{1/2+\alpha})$ for $\alpha>0$. Then the method described in this section, using an integration method explained in Lemma~\ref{lem:McDiamand}, provides with high probability an $\varepsilon$-approximation to the Holevo capacity with a complexity $O\big(M^6 \log(M \log M)^4(N+\log(M \log(M \log M))\big) \, 2^{c N^{1-\alpha}})=\subexp(N)$ for a constant $c>0$.
\end{enumerate}
\end{mycor}
\begin{proof}
Follows directly from Corollary~\ref{cor:complexity}.
\end{proof}

The following example presents families of channels $\{\Phi\}_{N,M}$ with an arbitrarily scaling Lipschitz constant $L_{N,M}$.
\begin{myex}[Familiy of channels with an arbitrary Lipschitz constant]
Consider the family of channels $\{\Phi : \mathcal{B}(\mathcal{H}_A) \to \mathcal{B}(\mathcal{H}_B)\}_{N,M}$ that maps $\rho \mapsto (1-\phi(N,M))\tfrac{\mathbf{1}}{M}+\phi(N,M) \Theta(\rho)$, where $\Theta : \mathcal{B}(\mathcal{H}_A) \to \mathcal{B}(\mathcal{H}_B)$ denotes an arbitrary cptp map and $\phi(N,M)\geq 0$. Following the lines of the proof of Lemma~\ref{lem:assLip} using that $\norm{\Phi(\rho_1)-\Phi(\rho_2)}_{\trnorm}=\phi(N,M)\norm{\Theta(\rho_1)-\Theta(\rho_2)}_{\trnorm}\leq\phi(N,M) \norm{\rho_1-\rho_2}_{\trnorm}$ it follows that the Lipschitz constant $L_{N,M}$  with respect to the $\ell^1$-norm of the function $f_{\lambda,M}$ as defined in Lemma~\ref{lem:assLip}  if given by $L_{N,M}=(2 M \sqrt{M}(\log(\tfrac{1}{\gamma_M} \vee \e)) + 2M (\log(\tfrac{1}{\gamma_M \e} \vee \e))) \phi(N,M)$.
\end{myex}

%%%%%%%%%%%%%%%%%%%%%%%%%%%%%%%%%%%%%%
\subsubsection*{Second approach: importance sampling} \label{subsec:algoForGradient:second}
The second approach invokes a non-trivial sampling method, known as \emph{importance sampling} \cite{ref:robert-04}.
Define the function $f_{\lambda}(x):=\Tr{\Phi(\Eu(x))\lambda}-H(\Phi(\Eu(x)))$ such that the gradient of $G_{\nu}(\lambda)$, given in \eqref{eq:gradient:G:no:input:constr}, can be expressed as
\begin{equation*}
\nabla G_{\nu}(\lambda) = \frac{\int_{R}2^{\frac{1}{\nu}f_{\lambda}(x)}\Phi(\Eu(x))\drv x}{\int_{R}2^{\frac{1}{\nu}f_{\lambda}(x)}\drv x} = \EQ{\Phi(\Eu(x))},
\end{equation*}
where the expectation is with respect to the probability density $Q(x)=\tfrac{2^{\frac{1}{\nu}f_{\lambda}(x)}}{\int_{R}2^{\frac{1}{\nu}f_{\lambda}(x)}\drv x}$. Consider i.i.d. random variables $\{X_{i}\}_{i=1}^{n}$ according to the density $Q$ and define the random variable $Z_{n}:=\tfrac{1}{n}\sum_{i=1}^{n}\Phi(\Eu(X_i))$.
\begin{mylem} \label{lem:second:approach:num:int}
For every $t\geq 0$ and $n\in\mathbb{N}$, $ \Prob{\norm{\nabla G_{\nu}(\lambda)-Z_{n}}_{\opnorm}\geq t}\leq M \exp\left( \frac{-t^{2}n}{32} \right)$.
\end{mylem}
\begin{proof}
The function defined as $R^{n}\ni x\mapsto f(x_{1},\hdots,x_{n}):=\tfrac{1}{n}\sum_{i=1}^{n}\Phi(\Eu(x_i))$ satisfies the following bounded difference assumption
\begin{align}
\norm{\left( f(x_{1},\hdots,x_{i},\hdots,x_{n})-f(x_{1},\hdots,x_{i-1},x_{i^\prime},x_{i+1},\hdots,x_{n}) \right)^{2}}_{\opnorm}
&\!\! \leq \! \left( \frac{1}{n} (\Phi(\Eu(x_i))\!-\!\Phi(\Eu(x_{i^\prime}))) \! \right)^{2}\label{step1:bdd:ass}\\
& \leq \frac{4}{n^{2}},\label{step2:bdd:ass}
\end{align}
where where \eqref{step1:bdd:ass} follows from $\norm{(B-C)^{2}}_{\opnorm}=\norm{B^{2}-BC-CB-C^{2}}_{\opnorm}\leq \norm{B^{2}}_{\opnorm}+\norm{BC}_{\opnorm}+\norm{CB}_{\opnorm}+\norm{C^{2}}_{\opnorm}\leq \norm{B}_{\opnorm}^{2}+2\norm{B}_{\opnorm}\norm{C}_{\opnorm}+\norm{C}_{\opnorm}^{2}=(\norm{B}_{\opnorm}+\norm{C}_{\opnorm})^{2}$ which uses the submultiplicative property of the operator norm. Inequality \eqref{step2:bdd:ass} is due to the fact that $\Phi(\Eu(x))$ are density operators for all $x\in R$. Hence, by the matrix Mc~Diarmid inequality \cite[Cor.~7.5]{ref:Tropp-12}, we get the concentration bound
\begin{equation*}
\Prob{\norm{\nabla G_{\nu}(\lambda)-Z_{n}}_{\opnorm}\geq t}\leq M \exp\left( \frac{-t^{2}n}{32} \right).
\end{equation*}
\end{proof}
The main difficulty in this approach is how to obtain samples $\{X_{i}\}_{i=1}^{n}$ according to the density $Q$ given above and in particular quantifying its computational complexity. It is well known that if the density $Q$ has a particular structure this samples can be drawn efficiently, e.g., if $Q$ is a log-concave density in polynomial time \cite{ref:Vempala-07}. 
Providing assumptions on the channel $\Phi$ such that sampling according to $Q$ can be done efficiently is a topic of further research.
\begin{myremark}
Let $\mathscr{S}(N,M)$ denote the computational cost of drawing one sample according to the density $Q$. Then, Lemma~\ref{lem:second:approach:num:int} shows that the computational complexity the gradient approximation given in Definition~\ref{def:oracle} using the importance sampling algorithm is $\mathscr{C}_{\Phi,\mathscr{A}}(N,M,\delta^{-1}, \eta^{-1}) = \mathscr{S}(N,M) \tfrac{32}{\delta^{2}}\ln\left(\tfrac{M}{\eta}\right) $. 
\end{myremark}

%%%%%%%%%%%%%%%%%%%%%%%%%%%%%%%%%%%%%%
%%%%%%%%%%%%%%%%%%%%%%%%%%%%%%%%%%%%%%
%%%%%%%%%%%%%%%%%%%%%%%%%%%%%%%%%%%%%%
\subsection{Simulation results}
The following three examples show the performance of our method to compute the Holevo capacity. In the first example we have chosen a quantum channel for which an analytical expression of the Holevo capacity is known. In the second example we demonstrate how to compute the classical capacity of an arbitrary qubit Pauli channel. As a third example, we have chosen a random qubit-input qubit-output channel for which the Holevo capacity is unknown.

The Choi-Jamiolkowski representation ensures that every quantum channel $\Phi:\mathcal{B}(\mathcal{H}_A)\to \mathcal{B}(\mathcal{H}_B)$ can be written as
\begin{equation*}
\sigma_B = \Phi(\rho_A) = N \, \Trp{A}{\left( \mathcal{T}_A(\rho_A)\otimes \textnormal{id}_B \right) \tau_{AB}},
\end{equation*}
where $\mathcal{T}_A(\cdot)$ is the transpose mapping and $\tau_{AB}$ denotes a density operator that fully characterizes the quantum channel and that satisfies $\Trp{B}{\tau_{AB}}=\tfrac{1}{N} \mathbf{1}$. For the following examples we use this representation of the channel.

Note that our method works for arbitrary quantum channels having a finite input dimension. The reason we have chosen qubit channels is to save computation time. All the simulations in this section are performed on a 2.3 GHz Intel Core i7 processor with 8 GB RAM with Matlab. For the evaluation of the gradient $\nabla G_{\nu}$ that involves the computation of an integral over the domain $[0,\pi]\times [0,2\pi]$ we used a trapezoidal method with a grid having $100 \times 200$ points. 
\begin{myex}[Qubit depolarizing channel] \label{ex:depolarizingChannel}
We consider the depolarizing channel with input and ouput dimension $2$, that can be described by the map $\rho_A \to (1-p) \rho_A + p  \tfrac{1}{2}\mathbf{1}$, for $p\in[0,1]$. It Choi state is given by $\tau_{AB}=(1-p) \ket{\omega} \! \bra{\omega} + \tfrac{p}{4} \mathbf{1}$, where $\ket{\omega}$ denotes a maximally entangled state. The Holevo capacity of the depolarizing channel can be computed analytically being \cite[Thm.~19.4.2]{wilde_book}
\begin{equation}
C_{\mathcal{X}}(p) = 1+\left(1-\frac{p}{2} \right)\log \left(1-\frac{p}{2} \right)+\frac{p}{2} \log\left(\frac{p}{2} \right)=1-\Hb\!\left(\frac{p}{2}\right).\label{eq:HolevoDepolarizing}
\end{equation}
Table~\ref{tab:depolarizing} shows the performance of our algorithm for the task of approximating the Holevo capacity for the depolarizing channel with parameter $p=\tfrac{1}{3}$. According to \eqref{eq:HolevoDepolarizing} the precise value of the Holevo capacity is $C_{\mathcal{X}}(p=\tfrac{1}{4})=1-\Hb(\tfrac{1}{6}) \approx 0.3499775784$.

 \begin{table}[!htb]
\centering 
\caption{Holevo Capacity of a depolarizing channel with $p=\tfrac{1}{3}$.}
\label{tab:depolarizing}

  \begin{tabular}{c @{\hskip 3mm} | @{\hskip 3mm} c @{\hskip 2mm} c @{\hskip 2mm} c}
 Iterations  & 10 & $10^2$ & $10^3$ \\
 $\nu$ & 0.1602 & 0.0174 & 0.0018\\
 $C_{\mathcal{X},\textnormal{UB}}$ &  0.3603 & 0.3500 & 0.3500 \\
  $C_{\mathcal{X},\textnormal{LB}}$ & 0.3500 & 0.3500 & 0.3500 \\
 A posteriori error & 1.029$\cdot10^{-2}$ &3.401$\cdot10^{-5}$ &8.019$\cdot10^{-6}$  \\
 Time [s] & 414 & 4144 & 41578
  \end{tabular}
\end{table}
\end{myex}

\begin{myex}[Qubit Pauli channel]
Consider the general Pauli channel for an input and output dimension $2$, which can be described by the map $\rho_A \to (1-p_X - p_Y - p_Z) \rho_A + p_X X \rho_A X + p_Y Y \rho_A Y + p_Z Z \rho_A Z$, where $X$,$Y$,$Z$ denote the Pauli matrices and $p_X,p_Y,p_Z \in [0,1]$ such that $p_X+p_Y+p_Z \in [0,1]$. The Choi state $\tau_{AB}$ representing this channel can be computed to be
\begin{equation*}
\tau_{AB} = \frac{1}{2} \left( \begin{matrix} 
1-p_X - p_Y & 0 & 0 & 1-p_X-p_Y-2p_Z \\
0& p_X + p_Y & p_X-p_Y & 0\\
0& p_X-p_Y & p_X + p_Y & 0 \\
1-p_X-p_Y-2p_Z & 0 & 0 & 1-p_X-p_Y
\end{matrix} \right).
\end{equation*}
\end{myex}
King proved that the Holevo capacity is additive for product channels, under the condition that one of the channels is a \emph{unital} qubit channel, with the other completely arbitrary \cite{king02}.\footnote{Unital channels are channels that map the identity to the identity, i.e., $\Phi(\textnormal{id})=\textnormal{id}$.} As Pauli channels are unital channels, the Holevo capacity is therefore equal to the classical capacity for arbitrary Pauli qubit channels. For certain qubit Pauli channels an analytical formula for the Holevo capacity is known (cf.\ the depolarizing channel in Example~\ref{ex:depolarizingChannel}), however in general the Holevo capacity is unknown. Our method introduced above allows us to approximate the Holevo capacity. To demonstrate this we compute upper and lower bounds for the Holevo capacity of a qubit Pauli channel with $p_X=\tfrac{1}{7}$, $p_Y=\tfrac{1}{10}$ and $p_Z=\tfrac{1}{4}$ as shown in Table~\ref{tab:PauliChannel}.
 \begin{table}[!htb]
\centering 
\caption{Holevo Capacity of a qubit Pauli channel with $p_X=\tfrac{1}{7}$, $p_Y=\tfrac{1}{10}$ and $p_Z=\tfrac{1}{4}$.}
\label{tab:PauliChannel}

  \begin{tabular}{c @{\hskip 3mm} | @{\hskip 3mm} c @{\hskip 2mm} c @{\hskip 2mm} c}
 Iterations & 10 & $10^2$ & $10^3$ \\
  $\nu$ & 0.1265 & 0.0138 & 0.0014 \\
 $C_{\mathcal{X},\textnormal{UB}}$ & 0.2026 & 0.2002 & 0.2002 \\
  $C_{\mathcal{X},\textnormal{LB}}$ & 0.1399 & 0.1894 &  0.1983\\
 A posteriori error & 6.267$\cdot10^{-2}$ &1.087$\cdot10^{-2}$ &1.940$\cdot10^{-3}$ \\
 Time [s] &  409 & 3919 & 40154 
  \end{tabular}
\end{table}

\begin{myex}[Random qubit channel]
We consider a random qubit-input qubit-output channel $\Phi:\mathfrak{T}(\mathcal{H}_A)\to \mathfrak{T}(\mathcal{H}_B)$ with $N=\dim(\mathcal{H}_A)=d_B=\dim(\mathcal{H}_B)=2$. More precisely, we consider the Choi state of $\Phi$, which is given by
\begin{equation*}
\tau_{AB} = \frac{1}{N} (\rho_A^{-\frac{1}{2}}\otimes \textnormal{id}_B) \, \rho_{AB} \, (\rho_A^{-\frac{1}{2}}\otimes \textnormal{id}_B),
\end{equation*}
where $\rho_{AB}$ is a random density matrix.\footnote{There are different methods to generate random density matrices which is however not relevant for this work. The interested reader might consider \cite{zyczkowski11} for further information.} To demonstrate the performance of our method, let
\begin{equation}
\tau_{AB}=\left( \begin{matrix} 
0.2041  & -0.1145 - 0.0926\ii &0.0590-0.0187\ii & 0.0721 + 0.0487\ii\\ 
 -0.1145 + 0.0926\ii & 0.2959 & -0.0861 - 0.0928\ii & -0.0590 + 0.00187\ii \\
 0.0590 + 0.0187\ii & -0.0861 + 0.0928\ii & 0.2350 & -0.1296 + 0.0128\ii \\
 0.0721 - 0.0487\ii & -0.0590 - 0.0187\ii & -0.1296 - 0.0128\ii & 0.2650 \\
\end{matrix} \right). \label{eq:exChannel}
\end{equation}

\end{myex}

 \begin{table}[!htb]
\centering 
\caption{Holevo Capacity of a random qq-channel described by its Choi state given in \eqref{eq:exChannel}.}
\label{tab:randomQQ}

  \begin{tabular}{c @{\hskip 3mm} | @{\hskip 3mm} c @{\hskip 2mm} c @{\hskip 2mm} c}
 Iterations  & 10 & $10^2$ & $10^3$ \\
  $\nu$ &  0.2575 & 0.0280 & 0.0028 \\
 $C_{\mathcal{X},\textnormal{UB}}$ &  0.3928 & 0.2648 & 0.2573 \\
  $C_{\mathcal{X},\textnormal{LB}}$  & 0.0900 & 0.2032 & 0.2522 \\
 A posteriori error &  3.028$\cdot10^{-1}$ &6.156$\cdot10^{-2}$ &5.061$\cdot10^{-3}$ \\
 Time [s] & 421 & 4025 & 41630 
  \end{tabular}
\end{table}

%%%%%%%%%%%%%%%%%%%%%%%%%%%%%%%%%
%%%%%%%%%%%%%%%%%%%%%%%%%%%%%%%%%

%%%%%%%%%%%%%%%%%%%%%%%%%%%%%%%%
%%%%%%%%%%%%%%%%%%%%%%%%%%%%%%%%
\section{Conclusion and Future Work} \label{sec:conclusion}
We have presented a new approach to approximate the capacity of cq channels with discrete or continuous bounded input alphabets possibly having constraints on the input distribution. More precisely, we derived iterative upper and lower bounds for the capacity and proved that they converge with a given rate. The dual problem of the cq channel capacity formula turns out to have a particular structure such that the Lagrange dual function admits a closed form solution. Applying smoothing techniques to the dual function allows us to finally approximate the problem efficiently. For cq channels with a discrete input alphabet of size $N$ and without additional input constraints, the complexity of generating an $\varepsilon$-close solution is $O(\tfrac{(N \vee M) M^3 (\log(N))^{1/2}}{\varepsilon})$ where $M$ denotes the output dimension. Using the idea of a universal encoder then enables us to extend the idea for the task of approximating the Holevo capacity. It turns out that the problem gets mapped to a multidimensional integration problem. We compute the complexity for generating an $\varepsilon$-close solution to the Holevo capacity using the new method. In addition, we derive assumptions on the family of channels under which an $\varepsilon$-close solution can be determined in subexponential or even polynomial time.

Recall that the classical capacity of a quantum channel $\Phi:\mathcal{B}(\mathcal{H}_A)\to \mathcal{B}(\mathcal{H}_B)$ is given by its regularized Holevo capacity, i.e., 
\begin{equation}
C(\Phi)= \lim_{k\to \infty} \frac{1}{k} C_{\mathcal{X}}(\Phi^{\otimes k}). \label{eq:HSWcapacity}
\end{equation}
The regularization required in \eqref{eq:HSWcapacity} makes the classical capacity of a quantum channel very hard to compute. If for some channel $\Phi$ the Holevo capacity is additive, i.e., $C_{\mathcal{X}}(\Phi \otimes \Theta)=  C_{\mathcal{X}}(\Phi) + C_{\mathcal{X}}(\Theta)$ for an arbitrary channel $\Theta$, this implies that $C(\Phi)=C_{\mathcal{X}}(\Phi)$ making the classical capacity a lot simpler to compute and proves that entangled states at the encoder do not help to improve the rate.

For a while there existed a conjecture that the Holevo capacity is additive for all quantum channels. In 2009 using techniques from measure concentration, Hastings disproved the conjecture by constructing high dimensional random quantum channels whose Holevo capacity is provably not additive \cite{hastings09}.
However, it remains unsolved whether there exist explicit small dimensional quantum channels whose Holevo capacity is not additive. Our approximation scheme can be used to check the additivity of the Holevo capacity for channels with small dimensions. 

%\ds{For example, we verified the Holevo capacity is additive for some qubit channels that are listed in Appendix~\ref{ap.}.}

The number of iterations our approximation scheme needs for an $\varepsilon$-solution highly depends on the Lipschitz constant estimate of the objective's gradient. Recently there has been some work motivating an adaptive estimate of the local Lipschitz constant that has been shown to be very efficient in practice (up to three orders of magnitude reduction of computation time), while preserving the worst-case complexity \cite{ref:Baes-14}. This may help to achieve a faster convergence for our algorithm, i.e., a smaller number of iterations would be required to achieve a certain approximation error.

Another idea to reduce the computation time of the approximation scheme is to make use of possible symmetry properties the channel might have. More precisely, certain symmetry properties could enable us to restrict the set $R$ over which one has to integrate in order to evaluate the gradient $\nabla G_{\nu}$. This would speed up the computational cost per iteration considerably. 

A different topic that deserves further investigation is to check whether the approach to approximate a capacity formula via smoothing its dual program might be applicable for different capacities such as the classical entanglement-assisted capacity or the channel coherent information.

%%%%%%%%%%%%%%%%%%%%%%%%%%%%%%%%%%

%%%%%%%%%%%%%%%%%%%%%%%%%%%%%%%%%%%

%%%%%%%%%%%%%%%%%%%%%%%%%%%%%%%%%%
%%%%%%%%%%%%%%%%%%%%%%%%%%%%%%%%%%%

%%%%%%%%%%%%%%%%%%%%%%%%%%%%%%%%%
\appendix

\section{Proof of Lemma~\ref{lem:inputConstraint}} \label{app:inputCost}
This proof follows a very similar structure as the proof of Lemma~2.1 in \cite{TobiasSutter14}. Adding the constraint $\sigma:=\sum_{i=1}^N p_i \rho_i$ gives $\I{p}{\rho}=\Hh{\sigma}-\sum_{i=1}^N p_i \Hh{\rho_i}$. Since $p \in \Delta_N$ and $\rho_i \in \mathcal{D}(\mathcal{H})$ for all $1\leq i \leq N $ it follows that $\sigma \in \mathcal{D}(\mathcal{H})$. 

By definition of $S_{\max}$ it is clear that the constraint $\inprod{p}{s} \leq S$ is inactive if $S\geq S_{\max}$ proving \eqref{eq:cqCapacityPrimalwithoutConst}. It remains to show that for $S < S_{\max}$ the optimization problems \eqref{eq:cqCapacity} and \eqref{eq:cqCapacityPrimal} are equivalent. To keep notation simple, let $C_{\mathsf{cq}}(S):=C_{\mathsf{cq},S}(\W)$ for some fixed cq channel $\W$. We next prove that $C_{\mathsf{cq}}(S)$ is concave in $S$ for $S \in [0,S_{\max}]$. Let $S^{(1)},S^{(2)}\in [0,S_{\max}]$, $\lambda \in [0,1]$ and let $p^{(i)}$ be capacity achieving input distribution for $C_{\mathsf{cq}}(S^{(i)})$ with $i \in \{1,2 \}$. Let $p^{(\lambda)}:= \lambda p^{(1)} + (1-\lambda) p^{(2)}$, which gives
\begin{align*}
\inprod{s}{p^{(\lambda)}} &= \lambda \inprod{s}{p^{(1)}} + (1-\lambda)\inprod{s}{p^{(2)}} \\
&\leq \lambda S^{(1)} + (1-\lambda) S^{(2)} \\
&=: S^{(\lambda)} \in [0,S_{\max}].
\end{align*}
Using the fact that $p \mapsto \I{p}{\rho}$ is concave\footnote{This follows directly from the well known fact that $p \mapsto \Hh{p}$ is concave.} we obtain
\begin{align*}
\lambda C_{\mathsf{cq}}(S^{(1)}) + (1-\lambda)C_{\mathsf{cq}}(S^{(2)}) &= \lambda \I{p^{(1)}}{\rho} + (1-\lambda) \I{p^{(2)}}{\rho} \\
&\leq \I{p^{(\lambda)}}{\rho} \\
& \leq C_{\mathsf{cq}}(S^{(\lambda)}), 
\end{align*}
where the final inequality follows from \eqref{eq:cqCapacity}. 

$C_{\mathsf{cq}}(S)$ is clearly non-degreasing in $S$ as enlarging $S$ relaxes the input cost constraint. We next show that $C_{\mathsf{cq}}(S)$ is even strictly increasing in $S \in [0,S_{\max}]$. We first prove that for all $\varepsilon >0$,
\begin{equation}
C_{\mathsf{cq}}(S_{\max}-\varepsilon) < C_{\mathsf{cq}}(S_{\max}). \label{eq:triangles}
\end{equation}
Suppose $C_{\mathsf{cq}}(S_{\max}-\varepsilon) = C_{\mathsf{cq}}(S_{\max})$ and denote $C_{\mathsf{cq}}^{\star}=\max_{p \in \Delta_N} \I{p}{\rho}$. This implies that there exists a $\bar p \in \Delta_N$ such that $\I{\bar p}{\rho}=C_{\mathsf{cq}}^{\star}$  and $\inprod{\bar p}{s}= S_{\max}- \varepsilon$, which contradicts the definition of $S_{\max}$. 
Thus by concavity of $C_{\mathsf{cq}}(S)$ together with the the non-decreasing property and \eqref{eq:triangles} imply that $C_{\mathsf{cq}}(S)$ is strictly increasing in $S$.

Finally, assume that $C_{\mathsf{cq}}(S)$ is achieved for $p^{\star} \in \Delta_N$ such that $\inprod{p^{\star}}{s}= \bar S < S$. Then we have
\begin{equation*}
C_{\mathsf{cq}}(\bar S):= \left \lbrace \begin{array}{ll} 
\max \limits_p & \I{p}{\rho}\\
\st & \inprod{p}{s} \leq \bar S\\
& p \in \Delta_N
\end{array} \right. = \I{p^{\star}}{\rho} = C_{\mathsf{cq}}(S),
\end{equation*}
which is a contradiction as $C_{\mathsf{cq}}(S)$ is strictly increasing in $S \in [0,S_{\max}]$. \qed
%%%%%%%%%%%%%%%%%%%%%%%%%%%%%%%%%%%%%%%%%%%%%%%%%%%%
%%%%%%%%%%%%%%%%%%%%%%%%%%%%%%%%%%%%%%%%%%%%%%%%%%%%
%%%%%%%%%%%%%%%%%%%%%%%%%%%%%%%%%%%%%%%%%%%%%%%%%%%%

\section{Proof of Lemma~\ref{lem:compact:set:CQ}} \label{app:compact}
The proof is extending the ideas used to prove \cite[Lem.~2.4]{TobiasSutter14}.
Consider the following two optimization problems
\begin{align*} 
\mathsf{P}_{\beta}:\left\{ \begin{array}{ll}
	\max\limits_{p,\sigma,\varepsilon} 		& \Hh{\sigma} - \sum_{i =1}^N p_i \Hh{\rho_i} - \beta\varepsilon \\
			\st						& \norm{ \sum_{i=1}^N p_i \rho_i -\sigma }_{\opnorm} \leq \varepsilon   \\
			                                				& \inprod{p}{s} = S\\
			                                				& p \in\Delta_N,  \sigma \in \mathcal{D}(\mathcal{H}), \varepsilon \in \Rp
	\end{array}\right.
  \quad \textnormal{and} \quad
	\quad 
\mathsf{D}_{\beta}: \left\{ \begin{array}{ll}
			\min\limits_{\lambda} 		&F(\lambda) + G(\lambda) \\
			\st					& \norm{\lambda}_{\trnorm} \leq \beta \\
									& \lambda\in \Hop.
	\end{array}\right. 
\end{align*}
\begin{myclaim}
Strong duality holds between $\mathsf{P}_{\beta}$ and $\mathsf{D}_{\beta}$.
\end{myclaim}
\begin{proof}
According to the identity $\norm{ \sum_{i=1}^N p_i \rho_i -\sigma }_{\opnorm}=\max_{\norm{\lambda}_{\trnorm}\leq 1} \inprod{\lambda}{\sum_{i=1}^N p_i \rho_i -\sigma}_F$ \cite[p.~7]{holevo_book} the optimization problem $\mathsf{P}_{\beta}$ can be rewritten as
\begin{align*} 
\mathsf{P}_{\beta}:  \left\{ \begin{array}{ll}
			\max\limits_{p,\sigma} 			&\Hh{\sigma} - \sum_{i =1}^N p_i \Hh{\rho_i} + \min\limits_{\norm{\lambda}_{\trnorm}\leq \beta}\inprod{\lambda}{\sum_{i=1}^N p_i \rho_i -\sigma}_F\\
			\text{s.t. }					& \inprod{p}{s} = S \\
			 							& p \in\Delta_N,  \sigma \in \mathcal{D}(\mathcal{H}),
	\end{array} \right.
\end{align*}
whose dual program, where strong duality holds according to \citep[Proposition~5.3.1, p.~169]{ref:Bertsekas-09} is given by
\begin{align*}
 \quad  \left\{ \begin{array}{lll}
			\min\limits_{\norm{\lambda}_{\trnorm}\leq \beta}  &\max\limits_{p,\sigma} 		&\Hh{\sigma} - \sum_{i =1}^N p_i \Hh{\rho_i} +  \inprod{\lambda}{\sum_{i=1}^N p_i \rho_i -\sigma}_F \\
			&\text{s.t. } 				  &\inprod{p}{s} = S \\
								&	&   p \in\Delta_N,  \sigma \in \mathcal{D}(\mathcal{H}),
	\end{array}\right. 
\end{align*}
which clearly is equivalent to $\mathsf{D}_{\beta}$ with $F(\cdot)$ and $G(\cdot)$ as given in \eqref{eq:GandF}.
\end{proof}
We denote by $\varepsilon^{\star}(\beta)$ the optimizer of $\mathsf{P}_{\beta}$ with the respective optimal value $J^{\star}_{\beta}$. Note that for
\begin{align}  \label{eq:J(eps):CQ}
J(\varepsilon):= \left\{ \begin{array}{ll}
	\max\limits_{p,\sigma} 		& \Hh{\sigma} - \sum_{i=1}^N p_i \Hh{\rho_i}  \\
			\st 						& \norm{ \sum_{i=1}^N p_i \rho_i -\sigma}_{\opnorm} \leq \varepsilon \\
			                                				& \inprod{p}{s} = S\\
			                                				& p \in\Delta_N,  \sigma \in \mathcal{D}(\mathcal{H})
	\end{array}\right. ,
	\end{align}
the mapping $\varepsilon \mapsto J(\varepsilon)$, the so-called perturbation function, is concave \cite[p.~268]{ref:BoyVan-04}. In a next step we write the optimization problem \eqref{eq:J(eps):CQ} in another equivalent form
\begin{align}  \label{eq:J(eps):equiv:CQ}
J(\varepsilon)=  \left\{ \begin{array}{ll}
			\max\limits_{p,v} 		&- \sum_{i=1}^N p_i \Hh{\rho_i} + \Hh{ \sum_{i=1}^N p_i \rho_i + \varepsilon v}  \\
			\st				&\norm{v}_{\opnorm} \leq 1 \\
								& \inprod{p}{s} = S \\
			 					& p\in\Delta_{N}, \ v\in\Hop.
	\end{array} \right. 
	\end{align}
The main idea of the proof is to show that for a sufficiently large $\beta$, which we will quantify in the following, the optimizer $\varepsilon^{\star}(\beta)$ of $\mathsf{P}_{\beta}$ is equal to zero.
That is, in light of the duality relations, the constraint $\norm{\lambda}_{\trnorm} \leq \tfrac{\beta}{2}$ in $\mathsf{D}_{\beta}$ is inactive and as such $\mathsf{D}_{\beta}$ is equivalent to $\mathsf{D}$. By using Taylor's theorem, there exists a $y_{\varepsilon}\in[0,\varepsilon]$ such that the entropy term in the objective function of \eqref{eq:J(eps):equiv:CQ} can be bounded as
\begin{align}
\Hh{ \sum_{i=1}^N p_i \rho_i + \varepsilon v}	&=		\Hh{ \sum_{i=1}^N p_i \rho_i} - \inprod{\log\left( \sum_{i=1}^N p_i\rho_{i}  \right) + \tfrac{1}{\ln 2} \boldsymbol{1}}{v}_{\!\!F}\varepsilon  \nonumber\\ 
 		& \qquad -\inprod{\left(  \sum_{i=1}^N p_i\rho_{i} + y_{\varepsilon} v \right)^{-1}}{v^{2}}_{\!\!F}\varepsilon^{2}\tfrac{1}{\ln 2}   \nonumber\\
		&\leq		\Hh{ \sum_{i=1}^N p_i \rho_i} - \inprod{\log\left( \sum_{i=1}^N p_i \rho_{i}  \right) +\tfrac{1}{\ln 2} \boldsymbol{1}}{v}_{\!\!F}\varepsilon + \frac{M}{\gamma \ln 2}\varepsilon^{2}.  \label{eq:Taylor:CQ}		
		\end{align}
Thus, the optimal value of problem $\mathsf{P}_{\beta}$ can be expressed as
\begin{align}
J_{\beta}^{\star} 	&= 		\max\limits_{\varepsilon} \left\{ J(\varepsilon)-\beta \varepsilon \right\} \nonumber \\
			&\leq		\max\limits_{\varepsilon} \left\{ \max\limits_{p,v}\left[ - \sum_{i=1}^N p_i \Hh{\rho_i} + \Hh{ \sum_{i=1}^N p_i \rho_i } \right. \right. \nonumber \\
			& \qquad \left. \left. - \inprod{\log\left( \sum_{i=1}^N p_i\rho_{i}  \right) + \tfrac{1}{\ln 2} \boldsymbol{1}}{v}_{\!\!F}\varepsilon \ : \ \inprod{p}{s} = S \right] + \frac{M}{\gamma \ln 2}\varepsilon^{2}  -\beta \varepsilon \right\} \label{eq:proof:compact:first:step:CQ}\\
			&\leq		\max\limits_{\varepsilon} \left\{ \max\limits_{p,v}\left[ - \sum_{i=1}^N p_i \Hh{\rho_i} + \Hh{ \sum_{i =1}^N p_i \rho_i } \ : \ \inprod{p}{s} = S \right] \right. \nonumber\\ 
			& \qquad  + (\rho - \beta)\varepsilon +  \frac{M}{\gamma \ln 2}\varepsilon^{2} \bigg\} \label{eq:proof:compact:second:step:CQ}\\
			&=		J(0) + \max\limits_{\varepsilon} \left\{  (\rho - \beta)\varepsilon\ +  \frac{M}{\gamma \ln 2}\varepsilon^{2} \right\}, \label{eq:proof:compact:third:step:CQ}
\end{align}
where $\rho = M \left( \log(\gamma^{-1}) \vee \tfrac{1}{\ln 2} \right)$. Note that \eqref{eq:proof:compact:first:step:CQ} follows from $\eqref{eq:J(eps):equiv:CQ}$ and \eqref{eq:Taylor:CQ}. The equation \eqref{eq:proof:compact:second:step:CQ} uses the fact that $- \inprod{\log\left( \sum_{i=1}^N p_i \rho_{i}  \right) + \tfrac{1}{\ln 2} \boldsymbol{1}}{v}_F \leq M \left( \log(\gamma^{-1}) \vee \tfrac{1}{\ln 2} \right)$. Thus, for $\beta>\rho$ and $\varepsilon_{1}=\tfrac{N\gamma}{M}(\rho-\beta)$, we get $\max\limits_{\varepsilon\leq \varepsilon_{1}} \left\{  (\rho - \beta)\varepsilon + \tfrac{M}{\gamma \ln 2}\varepsilon^{2} \right\} = 0.$ Therefore, \eqref{eq:proof:compact:third:step:CQ} together with the concavity of $\varepsilon$ implies that $J(0)$ is the global optimum of $J(\varepsilon)$ and as such $\varepsilon^{\star}(\beta)=0$ for $\beta>\rho$, indicating that $\mathsf{P}_{\beta}$ is equivalent to $\mathsf{P}$ in the sense that $J^{\star}_{\beta}=J^{\star}_{0}$. By strong duality this implies that the constraint $\norm{\lambda}_{\trnorm} \leq \beta$ in $\mathsf{D}_{\beta}$ is inactive. Finally, $\norm{\lambda}_{F}\leq\norm{\lambda}_{\trnorm}$ concludes the proof.  \qed

%%%%%%%%%%%%%%%%%%%%%%%%%%%%%%%%%%%%%%%%%%%%%%%%%%%%
\section{Proof of Proposition~\ref{prop:CQ:lipschitz}} \label{app:propLip}
The proof follows directly from the proof of Theorem 1 and Lemma 3 in \cite{nesterov05} together the following analysis.
Consider the operator $\WW: \mathcal{H}^{*}\to \R^{N}$ by $\WW\lambda:=\left(\inprod{\rho_{1}}{\lambda}_{F},\hdots,\inprod{\rho_{N}}{\lambda}_{F}\right)\transp$. Its operator norm can be bounded as
\begin{align}
\norm{\WW}_{\opnorm} 	&=		\max\limits_{\lambda\in \Hop, \ p\in\Delta_{N}} \left\{ \inprod{p}{\WW\lambda} \ : \norm{\lambda}_{F}=1, \ \norm{p}_{1}=1 \right\} \nonumber \\
			&\leq		\max\limits_{\lambda\in \Hop, \ p\in\Delta_{N}} \left\{ \left| \sum_{i=1}^{N}\inprod{\rho_{i}}{\lambda}_{F}p_{i} \right| \ : \norm{\lambda}_{F}=1, \ \norm{p}_{1}=1 \right\} \nonumber\\
			&\leq		\max\limits_{\lambda\in \Hop, \ p\in\Delta_{N}} \left\{ \sum_{i=1}^{N}\left|  \inprod{\rho_{i}}{\lambda}_{F}\right| p_{i} \ : \norm{\lambda}_{F}=1, \ \norm{p}_{1}=1 \right\} \label{eq:cq:classical:operator:triangle} \\
			&\leq		\max\limits_{p\in\Delta_{N}} \left\{ \sum_{i=1}^{N} \norm{\rho_{i}}_{F} p_{i} \ : \norm{\lambda}_{F}=1, \ \norm{p}_{1}=1 \right\} \label{eq:cq:classical:operator:CS} \\
			&\leq		1, \nonumber
\end{align}
where \eqref{eq:cq:classical:operator:triangle} follows from the triangle inequality, \eqref{eq:cq:classical:operator:CS} from Cauchy Schwarz and the last step is due to the fact that
\begin{align}
\| \rho_{i} \|_{F} 	&\leq 	\| \sqrt{\rho_{i}} \|_{F} \| \sqrt{\rho_{i}} \|_{F}  \label{eq:submultiplicative}\\
					&=		\sqrt{\Tr{\sqrt{\rho_{i}}\sqrt{\rho_{i}}\herm}} \sqrt{\Tr{\sqrt{\rho_{i}}\sqrt{\rho_{i}}\herm}} \nonumber\\
					&=		\sqrt{\Tr{\sqrt{\rho_{i}}\sqrt{\rho_{i}\herm}}} \sqrt{\Tr{\sqrt{\rho_{i}}\sqrt{\rho_{i}\herm}}} \label{eq:pos:herm:interchange}\\
					&=		\sqrt{\Tr{\rho_{i}}} \sqrt{\Tr{\rho_{i}}} \label{eq:property:density:operator1} \\
					&=		1, \label{eq:property:density:operator2}
\end{align}
where \eqref{eq:submultiplicative} is due to the submultiplicative property of the Frobenius norm and \eqref{eq:pos:herm:interchange} follows from the fact that $\rho_{i}$ is positive semi-definite. Finally, \eqref{eq:property:density:operator1} and \eqref{eq:property:density:operator2} follow since $\rho_{i}$ is a density operator.\qed
%%%%%%%%%%%%%%%%%%%%%%%%%%%%%%%%%%%%%%%%%%%%%%%%%%%%
\section{Proof of Proposition~\ref{prop:Lipschitz:continuity}} \label{app:propLipCont}
It is known, according to Theorem~5.1 in \cite{ref:devolder-12}, that $G_{\nu}(\lambda)$ is well defined and continuously differentiable at any $\lambda\in Q$ and that  this function is convex and its gradient $\nabla G_{\nu}(\lambda)=\WW^{\star} p_{\nu}^{\lambda}$ is Lipschitz continuous with constant $L_{\nu} =\tfrac{1}{\nu}\norm{\WW}^{2}$, where we have also used Lemma~\ref{lem:strong:convexity:CQ:cts}.
The operator norm can be simplified to 
\begin{align}
\norm{\WW}_{\opnorm} : 	&=		\sup\limits_{\lambda\in\Hop\!, \, p\in \Lp{1}(R)}\left\{ \inprod{p}{\WW\lambda} \ : \ \norm{\lambda}_{F}=1, \ \norm{p}_{1}=1 \right\} \nonumber \\
				&\leq		 \sup\limits_{\lambda\in\Hop\!, \, p\in \Lp{1}(R)}\left\{  \left| \int_{R} \Tr{\rho_{x}\lambda}\, p(x) \drv x \right| : \ \norm{\lambda}_{F}=1, \ \norm{p}_{1}=1 \right\}\nonumber \\
				&\leq		 \sup\limits_{\lambda\in\Hop\!, \, p\in \Lp{1}(R)}\left\{  \int_{R}  \left|\Tr{\rho_{x}\lambda}\right|\, p(x) \drv x  : \ \norm{\lambda}_{F}=1, \ \norm{p}_{1}=1 \right\}  \label{eq:lipschitz:proof:triangle}\\
				&\leq		 \sup\limits_{\lambda\in\Hop\!, \, p\in \Lp{1}(R)}\left\{  \int_{R} \norm{\rho_{x}}_{F}\, p(x) \drv x  : \ \norm{p}_{1}=1 \right\} \label{eq:lipschitz:proof:CS}\\
				&\leq		 \sup\limits_{p\in \Lp{1}(R)}\left\{  \int_{R}  p(x) \drv x  : \ \norm{p}_{1}=1 \right\}  \label{eq:lipschitz:proof:refer:to:finite:dim}\\
				&\leq 	1,  \nonumber 
\end{align}
where \eqref{eq:lipschitz:proof:triangle} follows from the triangle inequality, \eqref{eq:lipschitz:proof:CS} from Cauchy-Schwarz and \eqref{eq:lipschitz:proof:refer:to:finite:dim} is due to \eqref{eq:property:density:operator2}.  \qed
%%%%%%%%%%%%%%%%%%%%%%%%%%%%%%%%%%%%%%%%%%%%%%%%%%%%%%%%%%%%%%%%%%%%%%%%%%%%%%%%%%%%%%%%%%%%%%%%%%%%%%%%

%%%%%%%%%%%%%%%%%%%%%%%%%%%%%%%%%%%%%%%%%%%%%%%%%%%%%%%%%%%%%%%%%%
\section{Justification of Remark~\ref{remark:stopping2}} \label{app:lemma1}
\begin{mylem}\label{lem:lemma1:stop:proof}
For $\alpha\in\Rp$, consider the function $\Rps\ni \nu \mapsto \iota(\nu):=\nu \left( \log\nu^{-1} + \alpha\right)\in\R$. For all $\varepsilon\in\left( 0,2^{\alpha}\left( 1+ \tfrac{\log \e}{\e}\right)\right)$ if $\nu\leq \tfrac{\varepsilon}{\left( 1+ \tfrac{\log \e}{\e} \right)\left( \alpha + \log\left( \left( 1+ \tfrac{\log \e}{\e}\right)\varepsilon^{-1} \right) \right)}$, then $\iota(\nu)\leq \varepsilon$.
\end{mylem}
\begin{proof}
Note that for all $\bar{\varepsilon}\in(0,1)$
\begin{align*}
\iota\left( \frac{2^{\alpha}\bar{\varepsilon}}{\log \bar{\varepsilon}^{-1}} \right) = 2^{\alpha}\bar{\varepsilon}\left( 1+ \frac{\log \log \bar{\varepsilon}^{-1}}{\log \bar{\varepsilon}^{-1}} \right) \leq 2^{\alpha} \bar{\varepsilon} \left( 1 +\frac{\log \e}{\e} \right),
\end{align*}
where the last step is due to the fact that $\tfrac{\log x}{x}\leq \tfrac{\log \e}{\e}$ for all $x\in\Rps$ is used. It then suffices to consider $\varepsilon:=2^{\alpha}\left( 1+\tfrac{\log \e}{\e} \right)\bar{\varepsilon}$.
\end{proof}

%%%%%%%%%%%%%%%%%%%%%%%%%%%%%%%%%%%%%%%%%%%%%%%%%%%%%%%%%%%%%%%%%%

\section{Proof of Lemma~\ref{lem:assLip}} \label{app:lem:assLip}
\begin{myclaim}\label{claim:basic}
The function $R \ni r \mapsto \ket{r} \in \mathbb{C}^N$ as given in Remark~\ref{rmk:universalEncoder} satisfies $\norm{\ket{r_1}-\ket{r_2}}_{1} \leq N\norm{r_1 - r_2}_{1}$.
\end{myclaim}
\begin{proof}
Using the simple fact that if $f,g:\R \to [0,1]$ are two Lipschitz continuous function with constant $L_g$ and $L_f$ then $f(\cdot) g(\cdot)$ is Lipschitz continuous with constant $L_f + L_g$ we get
\begin{align}
 \norm{\ket{r_1}-\ket{r_2}}_{1} = \sum_{i=1}^N \left| \ket{r_1}_i \ket{r_2}_i \right| \leq N \norm{r_1 - r_2}_{1}.
\end{align}
\end{proof}

\begin{myclaim} \label{claim:encoderFoo}
The function $\Delta_n \ni x \mapsto f(x) = x\, x\transp  \in \Rp^{n\times n}$ satisfies $\norm{f(x)-f(y)}_{\trnorm}\leq  2 \sqrt{n} \norm{x-y}_1$.
\end{myclaim}
\begin{proof}
Let $x,y \in \Delta_n$, then by Cauchy-Schwarz we find
\begin{align}
\norm{f(x)-f(y)}^2_F &= \norm{x\, x\transp-y\,y\transp}^2_F\nonumber\\
 &=\norm{x}_2^4 + \norm{y}_2^4 - 2\inprod{x}{y}^2 \nonumber\\
 &\leq \norm{x}_2^4 + \norm{y}_2^4 - 2\inprod{x}{y}^2 + 2 \norm{x}_2^2 \norm{y}_2^2 -2\inprod{x}{y}^2\nonumber \\
 & = \left(\norm{x}_2^2 + \norm{y}_2^2 \right)^2 - \left(2\inprod{x}{y} \right)^2\nonumber\\
 & = \left(\norm{x}_2^2 + \norm{y}_2^2 +2\inprod{x}{y} \right) \left(\norm{x}_2^2 + \norm{y}_2^2 -2\inprod{x}{y} \right)\nonumber\\
 &= \left(\norm{x}_2^2 + \norm{y}_2^2 +2\inprod{x}{y} \right) \norm{x-y}_2^2 \label{eq:paralello}\\
 & \leq 4 \norm{x-y}_2^2 \label{eq:assprob},
\end{align}
where \eqref{eq:paralello} uses the parallelogram identity and \eqref{eq:assprob} follows since by assumption we have $\norm{x}_2 \leq \norm{x}_1 = 1$ and $\norm{y}_2 \leq \norm{y}_1 = 1$. For a matrix $A\in \R^{n\times n}$ the equivalence of the Frobenius and the trace norm \cite{ref:Ber-09}, i.e., $\norm{A}_F \leq \norm{A}_{\trnorm} \leq \sqrt{n} \norm{A}_F$ and the equivalence for vector norms, i.e., $\norm{x}_2 \leq \norm{x}_1 \leq \sqrt{n} \norm{x}_2$ for $x\in \R^n$ finally proves the assertion.
\end{proof}
%%%%%%%%%%%%%%%%%%%%%%%%%%%%%%%%
\begin{myclaim} \label{claim:entropyLip}
Let $\rho_1, \rho_2 \in \mathcal{D}(\mathcal{H})$ with $m=\dim \mathcal{H}$ and $c:=\min_{i \in \{1,2 \}} \min \spec \rho_i >0$. Then $\left|\Hh{\rho_1}- \Hh{\rho_2} \right| \leq L_m \norm{\rho_1-\rho_2}_{\trnorm}$ with $L_m:=\sqrt{m}(\log(\tfrac{1}{c \e}\vee \e))$.
\end{myclaim}
\begin{proof}
Consider the function $(0,1] \ni x \mapsto f(x)=-x \log x \in \Rp$. Note that $\tfrac{\partial f}{\partial x} = \log(\tfrac{1}{x e})$. As $f(\cdot)$ is a concave function we have for all $1\geq x_1\geq x_2 >0$, $f(x_2)-f(x_2) \leq \tfrac{\partial f}{\partial x}(x_1) (x_2-x_1)$. Thus it follows that $|f(x_1)-f(x_2)|\leq \max_{i \in \{1,2 \}} | \tfrac{\partial f}{\partial x}(x_i)| |x_1 -x_2|$ for all $x_1,x_2 \in (0,1]$, which then implies that for all $x_1,x_2 \in (0,1]$ and $c \in (0,1)$
\begin{equation}
|f(x_1)-f(x_2)| \leq \left( \log(\tfrac{1}{c e}) \vee \log(e)\right) |x_1 - x_2|. \label{eq:stepPeyman}
\end{equation}
For $\rho_1, \rho_2 \in \mathcal{D}(\mathcal{H})$, let $\spec(\rho_1) = \{\lambda_1^{(1)}, \lambda_1^{(2)},\ldots, \lambda_1^{(m)}Ê\}$ and $\spec(\rho_2) = \{\lambda_2^{(1)}, \lambda_2^{(2)},\ldots, \lambda_2^{(m)}Ê\}$. Using the triangle inequality then gives
\begin{align}
\left| \Hh{\rho_1} - \Hh{\rho_2} \right| &= \left| \sum_{i=1}^m - \lambda_1^{(i)} \log(\lambda_1^{(i)}) + \lambda_2^{(i)}\log(\lambda_2^{(i)}) \right| \nonumber\\
 &\leq \sum_{i=1}^m \left| - \lambda_1^{(i)} \log(\lambda_1^{(i)}) + \lambda_2^{(i)}\log(\lambda_2^{(i)}) \right| \nonumber\\
 &= \sum_{i=1}^m \left| f(\lambda_1^{(i)})- f(\lambda_2^{(i)}) \right| \nonumber\\
 &\leq \left( \log\left(\frac{1}{c \e}\right) \vee \log\left(\e\right)\right) \sum_{i=1}^m \left| \lambda_1^{(i)} - \lambda_2^{(i)} \right| \label{eq:myassu}\\
 & \leq \left( \log\left(\frac{1}{c \e} \vee e\right)\right)  \sqrt{m} \left( \sum_{i=1}^m \left| \lambda_1^{(i)}-\lambda_2^{(i)} \right| \right)^{1/2} \label{eq:equivNorm}\\
 & \leq \left( \log\left( \frac{1}{c \e}\vee \e\right)\right)  \sqrt{m} \norm{\rho_1 - \rho_2}_F \label{eq:WH}\\
 & \leq \left( \log\left(\frac{1}{c \e} \vee \e\right)\right)  \sqrt{m} \norm{\rho_1 - \rho_2}_{\trnorm} \label{eq:finalDA},
\end{align}
where \eqref{eq:myassu} follows by assumption together with \eqref{eq:stepPeyman}. Inequality \eqref{eq:equivNorm} uses the equivalence of the one and two vector norm and that the logarithm is monotonic. Inequality \eqref{eq:WH} uses the Hoffman-Wielandt inequality \cite[p.~56]{tao12}. Finally, \eqref{eq:finalDA} follows from the equivalence of the Frobenius and the trace norm.
\end{proof}
%%%%%%%%%%%%%%%%%%%%%%%%%%%%%%%%
For $x_1, x_2 \in R$, the triangle inequality gives
\begin{align}
|f_{\lambda,M}(x_1)-f_{\lambda,M}(x_2)|&=\left|\Tr{\Phi(\Eu(x_1))\lambda}- \Hh{\Phi(\Eu(x_1))} -\Tr{\Phi(\Eu(x_2))\lambda} +  \Hh{\Phi(\Eu(x_2))}\right|\nonumber\\
&\leq \left|\inprod{\Phi(\Eu(x_1))}{\lambda}_F\!\!- \inprod{\Phi(\Eu(x_2))}{\lambda}_F  \right|\! + \!\left| \Hh{\Phi(\Eu(x_1))}\!-\!\Hh{\Phi(\Eu(x_2))} \right|. \label{eq:tbound}
\end{align}
Using Cauchy-Schwarz and the linearity of quantum channels we can bound the first part of \eqref{eq:tbound} as
\begin{align}
\left|\inprod{\Phi(\Eu(x_1))}{\lambda}- \inprod{\Phi(\Eu(x_2))}{\lambda}_F  \right| & = \left|\inprod{\Phi(\Eu(x_1)-\Eu(x_2))}{\lambda}_F\right|\nonumber\\
&\leq \norm{\Phi(\Eu(x_1)-\Eu(x_2))}_F \norm{\lambda}_F\nonumber\\
&\leq \norm{\Phi(\Eu(x_1)-\Eu(x_2))}_{\trnorm} \norm{\lambda}_F\label{eq:eq_norm}\\
&\leq \norm{\Eu(x_1)-\Eu(x_2)}_{\trnorm} \norm{\lambda}_F \label{eq:contraction}\\
&\leq 2 N \sqrt{N}\norm{x_1-x_2}_{1} \norm{\lambda}_F, \label{eq:finiii}
\end{align}
where \eqref{eq:eq_norm} uses the equivalence of the Frobenius and the trace norm \cite{ref:Ber-09} and inequality \eqref{eq:contraction} is a direct consequence of the contractivity property under the trace norm of quantum channels \cite[Thm.~8.16]{wolfScript}. Inequality \eqref{eq:finiii} follows from Claims~\ref{claim:basic} and \ref{claim:encoderFoo}. 

Recall that $ \norm{\lambda}_{F}\leq M \log\left(\gamma_M^{-1} \vee \e \right)$ as by definition $\lambda \in \Lambda$.

With the help of Claim~\ref{claim:entropyLip} and Assumption~\ref{ass:FPRAS} we can also bound the second part of \eqref{eq:tbound}. Let $J_M:=\sqrt{M}(\log(\tfrac{1}{\gamma_M \e} \vee \e))$ we then have
\begin{align}
\left| \Hh{\Phi(\Eu(x_1))}-\Hh{\Phi(\Eu(x_2))} \right| &\leq J_M \norm{\Phi(\Eu(x_1))-\Phi(\Eu(x_2))}_{\trnorm}\nonumber\\
&=J_M \norm{\Phi(\Eu(x_1)-\Eu(x_2))}_{\trnorm}\nonumber\\
& \leq J_M \norm{\Eu(x_1)-\Eu(x_2)}_{\trnorm} \label{eq:contract2}\\
& \leq 2N \sqrt{N} J_M\norm{x_1-x_2}_1 \label{eq:reallydone},
\end{align}
where \eqref{eq:contract2} again uses the contractivity property under the trace norm of quantum channels \cite[Thm.~8.16]{wolfScript} and \eqref{eq:reallydone} follows from Claims~\ref{claim:basic} and \ref{claim:encoderFoo}.
\qed

%%%%%%%%%%%%%%%%%%%%%%%%%%%%%%%%%%%%%%
\section{Proof of Lemma~\ref{lem:McDiamand}} \label{app:lem:mcdiarmid}
Within this proof we use the notation $\rho_x := \Phi(\Eu(x))$.
We define the functions $R\ni x\mapsto f_{\lambda}(x):=\WW\lambda(x) - \Hh{\rho_x} = \Tr{\rho_x \lambda } -\Hh{\rho_x}\in\R$ and $R\ni x\mapsto g_{\lambda}(x):=f_{\lambda}(x)-\bar{f}_{\lambda}\in\R_{\leq 0}$, where $\bar{f}_{\lambda}:=\max_{x\in R}f_{\lambda}(x)=f_{\lambda}(x^{\star})$. Then, by following Remark~\ref{remark:comp:stablility}, we have
\begin{equation*}
\nabla G_{\nu}(\lambda) = \frac{1}{\bar{S}(\lambda)} \int_{R} 2^{\tfrac{1}{\nu}g_{\lambda}(x)} (\rho_x\transp-\rho_{x^{\star}}\transp) \drv x + \rho_{x^{\star}}\transp, 
\end{equation*}
 where
 \begin{equation*}
 \bar{S}(\lambda)=\int_{R} 2^{\tfrac{1}{\nu} g_{\lambda}(x)} \drv x
 \end{equation*}
and we have used $\tfrac{\partial \Tr{\rho \lambda }}{\partial \lambda_{k,\ell}}=\rho_{\ell,k}$ \cite[Prop.~10.7.2]{ref:Ber-09}.
Consider i.i.d. random variables $\{X_{i}\}_{i=1}^{n}$ taking values in $R$. Define the random variable $\bar{S}_{n}(\lambda):=\tfrac{1}{n}\sum_{i=1}^{n}2^{\tfrac{1}{\nu}g_{\lambda}(X_{i})}$. Then, invoking the non-positivity of $g_{\lambda}(\cdot)$, Mc Diarmid's inequality \cite[Thm.~2.2.2]{ref:Raginsky-13} leads to the following concentration bound
\begin{equation}\label{eq:denominator}
\Prob{\left| \bar{S}(\lambda) - \bar{S}_{n}(\lambda) \right|\geq t}\leq 2\exp\left( -2t^{2}n \right).
\end{equation}
Next, we approximate $T(\lambda):= \int_{R} 2^{\tfrac{1}{\nu}g_{\lambda}(x)} (\rho_x\transp-\rho_{x^{\star}}\transp) \drv x$. Consider i.i.d. random variables $\{X_{i}\}_{i=1}^{n}$ taking values in $R$ and define a function $R^{n}\ni x \mapsto f(x_{1}, \dots, x_{n}):=\tfrac{1}{n}\sum_{i=1}^{n} 2^{\tfrac{1}{\nu}g_{\lambda}(x_{i})} (\rho_{x_{i}}\transp-\rho_{x^{\star}}\transp)\in \Hop$.
\begin{myclaim}
The function $f$ satisfies the bounded difference assumption 
\begin{align*}
&\sup\limits_{x_{1},\dots,x_{n},x_{i}^{\prime}}  \left( f(x_{1},\hdots,x_{i},\hdots,x_{n}) - f(x_{1},\hdots,x_{i-1},x_{i}^{\prime},x_{i+1},\hdots,x_{n}) \right)^{2} \\
&\hspace{15mm} \preccurlyeq \mathsf{diag}(\tfrac{3}{n},\dots,\tfrac{3}{n})^{2} \text{ for all }i=1,\hdots, n.
\end{align*}
\end{myclaim}
\begin{proof} 
\begin{align*}
&f(x_{1},\hdots,x_{i},\hdots,x_{n}) - f(x_{1},\hdots,x_{i-1},x_{i}^{\prime},x_{i+1},\hdots,x_{n})\\
& \quad =\frac{1}{n}\left( 2^{\tfrac{1}{\nu}g_{\lambda}(x_{i})}(\rho_{x_{i}}\transp-\rho_{x^{\star}}\transp)- 2^{\tfrac{1}{\nu}g_{\lambda}(x_{i}^{\prime})}(\rho_{x_{i}^{\prime}}\transp-\rho_{x^{\star}}\transp)\right) \\
& \quad = \frac{1}{n}\rho_{x^{\star}}\transp\left( 2^{\tfrac{1}{\nu}g_{\lambda}(x_{i}^{\prime})} - 2^{\tfrac{1}{\nu}g_{\lambda}(x_{i})}  \right) + \frac{1}{n}\left( 2^{\tfrac{1}{\nu}g_{\lambda}(x_{i})}\rho_{x_{i}}\transp - 2^{\tfrac{1}{\nu}g_{\lambda}(x_{i}^{\prime})}\rho_{x_{i}^{\prime}}\transp\right)\\
& \quad = \frac{1}{n}\left( \rho_{x^{\star}}\transp(b_{x_{i}^{\prime}}-b_{x_{i}})  + b_{x_{i}}\rho_{x_{i}}\transp - b_{x_{i}^{\prime}}\rho_{x_{i}^{\prime}}\transp \right) = : (\star),
\end{align*}
where $b_{y}:=2^{\tfrac{1}{\nu}g_{\lambda}(y)}$.
Now,
\begin{align}
\lambda_{\max}\left( (\star)^{2} \right) &= \norm{(\star)^{2}}_{\opnorm}\nonumber \\
		&\leq\left( \left| \frac{b_{x_{i}}}{n} - \frac{b_{x^{\prime}_{i}}}{n}\right| \norm{\rho_{x^{\star}}\transp}_{\opnorm} + \norm{  \frac{b_{x^{\prime}_{i}}}{n} \rho_{x^{\prime}_{i}}\transp - \frac{b_{x_{i}}}{n}\rho_{x_{i}}\transp }_{\opnorm} \right)^{2} \label{eq:step1:bla}\\
		&\leq \left( \frac{1}{n} \left| b_{x_{i}} - b_{x^{\prime}_{i}}\right| \norm{\rho_{x^{\star}}}_{\opnorm} + \frac{1}{n} \left|  b_{x^{\prime}_{i}}\right| \norm{\rho_{x^{\prime}_{i}}}_{\opnorm} + \frac{1}{n} \left|  b_{x_{i}}\right| \norm{\rho_{x_{i}}}_{\opnorm}   \right)^{2}\label{eq:step2:bla} \\
		&\leq \frac{9}{n^{2}}\label{eq:step3:bla}, 
\end{align}
where \eqref{eq:step1:bla} follows from $\norm{(B-C)^{2}}_{\opnorm}=\norm{B^{2}-BC-CB-C^{2}}_{\opnorm}\leq \norm{B^{2}}_{\opnorm}+\norm{BC}_{\opnorm}+\norm{CB}_{\opnorm}+\norm{C^{2}}_{\opnorm}\leq \norm{B}_{\opnorm}^{2}+2\norm{B}_{\opnorm}\norm{C}_{\opnorm}+\norm{C}_{\opnorm}^{2}=(\norm{B}_{\opnorm}+\norm{C}_{\opnorm})^{2}$ which uses the submultiplicativ property of the operator norm. Equation \eqref{eq:step2:bla} is due to the triangle inequality and \eqref{eq:step3:bla} uses the non-positivity of the function $g_{\lambda}$ and the property of density operators.
\end{proof}
Define the random variable $T_{n}(\lambda):=\tfrac{1}{n}\sum_{i=1}^{n} 2^{\tfrac{1}{\nu}g_{\lambda}(X_{i})} (\rho_{X_{i}}\transp-\rho_{x^{\star}}\transp)$
\begin{myclaim}\label{claim:nominator}
$\Prob{\norm{T_{n}(\lambda)-T(\lambda)}_{\opnorm}\geq t }\leq M \exp \left( \frac{-t^{2}n}{72} \right)$
\end{myclaim}
\begin{proof}
By the matrix Mc~Diarmid inequality \cite[Cor.~7.5]{ref:Tropp-12}, we get the concentration bound
\begin{equation*}
\Prob{\lambda_{\max}(T_{n}(\lambda)-T(\lambda))\geq t }\leq M \exp \left( \frac{-t^{2}n}{72} \right).
\end{equation*}
Furthermore, as pointed out in \cite[Rmk.~3.10]{ref:Tropp-12}, $\lambda_{\min}(X)=-\lambda_{\max}(-X)$. As such following similar lines as above one can derive
\begin{equation*}
\Prob{\lambda_{\min}(T_{n}(\lambda)-T(\lambda))\leq -t }\leq M \exp \left( \frac{-t^{2}n}{72} \right).
\end{equation*}
\end{proof}
%%%%%%%%%%%%%%%%%%%

\begin{myclaim} \label{claim:bounding:fraction}
Let $A,B \in \R$, $\xi_1,\xi_2 \geq 0$, $B > \xi_2$, $\hat A \in [A-\xi_1,A+\xi_1]$ and $\hat B \in [B-\xi_2,B+\xi_2]$. Then for $Z:=\tfrac{A}{B}$ and $\hat Z:=\tfrac{\hat A}{\hat B}$ we have
\begin{equation*}
\left| Z - \hat Z \right| \leq \max \left \lbrace  \frac{A}{B}- \frac{A-\xi_1}{B+\xi_2}, \frac{A+\xi_1}{B-\xi_2}-\frac{A}{B} \right \rbrace.
\end{equation*}
\end{myclaim}
\begin{proof}
Define 
\begin{equation*}
\hat Z_{\min}:=\frac{A-\xi_1}{B+\xi_2}  \quad \textnormal{and} \quad \hat Z_{\max}:=\frac{A+\xi_1}{B-\xi_2}
\end{equation*}
such that $\hat Z_{\min} \leq \hat Z \leq \hat Z_{\max}$. The inequality $|Z-\hat Z| \leq \max\{ Z-\hat Z_{\min},\hat Z_{\max}-Z \}$ finally proves the assertion.
\end{proof}

According to Claim~\ref{claim:bounding:fraction}, Equation~\eqref{eq:denominator} together with Claim~\ref{claim:nominator} give
\begin{align}\label{eq:concentration:bound:1}
\Prob{\norm{\nabla G(\lambda) - \nabla \tilde{G}(\lambda)}_{\opnorm}\geq \varphi(t)}\leq M \exp \left( \frac{-t^{2}n}{72} \right),
\end{align}
where $\varphi(t):=\max\left\{ \frac{(\norm{T(\lambda)}_{\opnorm}+\bar{S}(\lambda))t}{\bar{S}(\lambda)(\bar{S}(\lambda)-t)}, \frac{(\norm{T(\lambda)}_{\opnorm}+\bar{S}(\lambda))t}{\bar{S}(\lambda)(\bar{S}(\lambda)+t)} \right\}$. We next show that $\bar{S}(\lambda)$ is uniformly away from zero and restrict values of $t$ to an interval as such $\varphi$ well defined.
Recall that $x^{\star}\in R$ is such that $g_{\lambda}(x^{\star})=0$. Therefore
\begin{align} \label{eq:important bound on nominator}
\bar{S}(\lambda) 		&=		\int_{R}2^{\tfrac{1}{\nu}g_{\lambda}(x)}\drv x 
						\geq	\int_{\mathsf{B}_{\varepsilon}(x^{\star})\cap R}2^{\tfrac{1}{\nu}g_{\lambda}(x)}\drv x
						\geq \int_{\mathsf{B}_{\varepsilon}(x^{\star})\cap R}2^{\tfrac{-\sqrt{N}L_{N,M}\varepsilon}{\nu}} \drv x 
						 \geq 2^{\tfrac{-\sqrt{N}L_{N,M}\varepsilon}{\nu}} \varepsilon^{N},
\end{align}
where we have used the Lipschitz continuity of $g_{\lambda}$ given by Lemma~\ref{lem:assLip} with respect to the $\lp{\infty}$-norm  and considered the ball $\mathsf{B}_{\varepsilon}(x^{\star})$, centered at $x^{\star}$ with radius $\varepsilon$ with respect to the $\lp{\infty}$-norm. By choosing $\varepsilon=1$, one gets 
\begin{align*}
\bar{S}(\lambda) &\geq 2^{\frac{-\sqrt{N}L_{N,M}}{\nu}},
\end{align*}
which is strictly away from zero for any finite $N$. Moreover, the inequality \eqref{eq:concentration:bound:1} holds for all $t\in (0,2^{\frac{-\sqrt{N}L_{N,M}}{\nu}})$.
\begin{myclaim}
For $t\in[0,\tfrac{\bar{S}(\lambda)}{2}]$ and $\min\limits_{\lambda\in\Lambda}\tfrac{\bar{S}(\lambda)^{4}}{576}\geq \tfrac{1}{576} 2^{\frac{-4\sqrt{N}L_{N,M}}{\nu}}=:K_{N,M}$
\begin{equation*}
\Prob{\norm{\nabla G(\lambda) - \nabla \tilde{G}(\lambda)}_{\opnorm}\geq t}\leq M \exp \left( -K_{N,M} t^{2}n \right),
\end{equation*}
\end{myclaim}
\begin{proof} 
Define $\alpha_{\varepsilon}:=\tfrac{\norm{T(\lambda)}_{\opnorm}+\bar{S}(\lambda)}{\bar{S}(\lambda)(\bar{S}(\lambda)-\varepsilon)}$ and $\beta_{\varepsilon}:=\tfrac{\norm{T(\lambda)}_{\opnorm}+\bar{S}(\lambda)}{\bar{S}(\lambda)(\bar{S}(\lambda)+\varepsilon)}$. It can be seen that $\alpha_{\varepsilon}\geq\beta_{\varepsilon}$ for any $\varepsilon\in[0,\bar{S}(\lambda))$ and as such
\begin{align}
\varphi(t) \leq \alpha_{\varepsilon} t = \tfrac{2(\norm{T(\lambda)}_{\opnorm}+\bar{S}(\lambda))}{\bar{S}(\lambda)^{2}} t=:\alpha t \quad \text{for all } t\in[0,\varepsilon],
\end{align}
where we have chosen $\varepsilon=\tfrac{\bar{S}(\lambda)}{2}$. By \eqref{eq:concentration:bound:1} this gives
\begin{align*}
\Prob{\norm{\nabla G(\lambda) - \nabla \tilde{G}(\lambda)}_{\opnorm}\geq \alpha t}\leq M \exp \left( \frac{-t^{2}n}{72} \right),
\end{align*}
which shows that $K_{N,M}\geq \frac{\bar{S}(\lambda)^{4}}{288(\norm{T(\lambda)}_{\opnorm}+\bar{S}(\lambda))^{2}}$. Using $\norm{T(\lambda)}_{\opnorm}\leq 1$ and $\bar{S}(\lambda)\leq 1$
completes the proof.
\end{proof} 
\qed

%%%%%%%%%%%%%%%%%%%%%%%%%%%%%%%%%%%%%%%%%%%%%%%%%%%%%%%%%%%%%%%%%%%%%%%%%%%%%%%%%%%%%%%%%%%%%%%%%%%%%%%%
\section{Proof of Corollary~\ref{cor:complexity}} \label{app:cor:complexity}
This proof uses the same notation as the proof of Theorem~\ref{thm:FRRAS}.
For a fixed accuracy $\varepsilon >0$, Remark~\ref{rmk:avoidingAssFPRAS} implies that without loss of generality we can assume that $\log(\tfrac{1}{\gamma_M}) = \log( M \log M) =:\p(M)$.
Recall that as explained in the proof of Theorem~\ref{thm:FRRAS} the smoothing parameter $\nu$ is chosen as $\nu\leq \tfrac{\varepsilon}{3 \beta \left( \alpha + \log(3\beta \varepsilon^{-1})\right)}$ for $\beta:=1+\tfrac{\log \e}{\e}$ and $\alpha:=\log(L_{N,M}) + (2N-2)\log(2 \pi) +1$. It can be verified immediately that $\nu^{-1} = \Omega(N + \log(N^{3/2}M \p(M)))$.
Let $\delta = O(\tfrac{1}{M\p(M}2^{c \frac{\sqrt{N}}{\nu}L_{N,M}})$ for some constant $c>0$. According to Lemma~\ref{lem:McDiamand}, to ensure that $\eta^{-1}= \Omega(M^2\p(M)^2(N+\log(M \p(M))))$ we have to choose the number of samples as
\begin{equation}
n = O\left(M^2 \p(M)^2 2^{c' \frac{\sqrt{N}}{\nu}L_{N,M}} \right) = O\left(M^2 \p(M)^2 2^{c'(N^{3/2}+N^{1/2}\log(N^{3/2}M \p(M)))L_{N,M}}  \right), \label{eq:samplesNaive}
\end{equation}
for some constant $c'>0$. Note that the complexity to generate $n$ i.i.d. uniformly distributed samples $\{X_i\}_{i=1}^n$ is $O(n)$.  The total complexity to ensure an $\varepsilon$-close solution is then $k \, M^2 n$ with $k$ being the number of iterations that is given in \eqref{eq:iterationsFRPAS}. Recalling that $\p(M):= \log(M \log M)$ then proves the assertion. \qed
%%%%%%%%%%%%%%%%%%%%%%%%%%%%%%%%%%%%%%%%%%%%%%%%%%%%%%%%%%%%%%%%%%%%%%%%%%%%%%%%%%%%%%%%%%%%%%%%%%%%%%%%
\section*{Acknowledgments}
We would like to thank Omar Fawzi, John Lygeros and Stefan Richter for helpful discussions and pointers to references. We also thank Aram Harrow and Ashley Montanaro for sharing with us their vision and discernment on \cite{harrow13}.
DS and RR acknowledge support by the Swiss National Science Foundation (through the National Centre of Competence in Research `Quantum Science and Technology' and grant No.~200020-135048) and by the European Research Council (grant No.~258932).
TS and PME were supported by the ETH grant (ETH-15 12-2) and the HYCON2 Network of Excellence (FP7-ICT-2009-5).
%%%%%%%%%%%%%%%%%%%%%%%%%%%%%%%%%%%%%%%%%%%%%%%%%%%%
%%%%%%%%%%%%%%%%%%%%%%%%%%%%%%%%%%%%%%%%%%%%%%%%%%%%
%\bibliographystyle{alpha}
\bibliography{./bibtex/header,./bibtex/bibliofile}

%%%%%%%%%%%%%%%%%
%%%%%%%%%%%%%%%%%
%%%%%%%%%%%%%%%%%
%%%%%%%%%%%%%%%%%

\clearpage

%%%%%%%%%%%%%%%%%
%%%%%%%%%%%%%%%%%
%%%%%%%%%%%%%%%%%

%\widetext

\end{document}